\documentclass[11pt]{article}
\usepackage{microtype}
\usepackage{hyperref}
\usepackage{fullpage}

\usepackage{amsmath, amsthm, amssymb}
\usepackage{mathtools}
\usepackage{cite}
\usepackage{url}
\usepackage{appendix}
\usepackage{graphicx}
\usepackage{epstopdf}
\usepackage{setspace}
\usepackage{enumerate}
\usepackage{color}
\usepackage{xcolor}
\usepackage{algorithm}
\usepackage[noend]{algpseudocode}
\usepackage{multirow}
\usepackage{paralist}
\usepackage{lmodern}
\usepackage{microtype}

\usepackage{caption}
\usepackage{xspace}
\usepackage{mdwlist}
\usepackage{tikz}

\newboolean{short}
\setboolean{short}{false}
\newcommand{\onlyShort}[1]{\ifthenelse{\boolean{short}}{#1}{}}
\newcommand{\onlyLong}[1]{\ifthenelse{\boolean{short}}{}{#1}}

\newtheorem{theorem}{Theorem}[section]
\newtheorem{lemma}[theorem]{Lemma}
\newtheorem{claim}[theorem]{Claim}

\newtheorem{definition}[theorem]{Definition}

\newcommand{\calvin}[1]{ [[[ \textcolor{blue}{\bf Cal:} {\em #1} ]]]}
\newcommand{\sebastian}[1]{ [[[ \textcolor{green}{\bf Seb:} {\em #1} ]]]}
\newcommand{\mohsen}[1]{ [[[ \textcolor{blue}{\bf Mohsen:} {\em #1}
  ]]]}
\newcommand{\fabian}[1]{ [[[ \textcolor{orange}{\bf Fabian:} {\em #1}
  ]]]}
\newcommand{\seth}[1]{ [[[ \textcolor{red}{\bf Seth:} {\em #1} ]]]}

\newcommand{\hide}[1]{}

\renewcommand{\vec}[1]{\underline{#1}}

\newcommand{\brackets}[1]{\ensuremath{\left(#1\right)}}

\newcommand{\whp}[1][\empty]{\ensuremath{\text{w.h.p.}\ifthenelse{\equal{#1}{\empty}}{}{(#1)}}}
\newcommand{\Whp}[1][\empty]{\ensuremath{\text{W.h.p.}\ifthenelse{\equal{#1}{\empty}}{}{(#1)}}}

\renewcommand{\Pr}{\mathbb{P}}

\algnewcommand\algorithmicswitch{\textbf{switch}}
\algnewcommand\algorithmiccase{\textbf{case}}

\algdef{SE}[SWITCH]{Switch}{EndSwitch}[1]{\algorithmicswitch\ #1\ \algorithmicdo}{\algorithmicend\ \algorithmicswitch}%
\algdef{SE}[CASE]{Case}{EndCase}[1]{\algorithmiccase\ #1}{\algorithmicend\ \algorithmiccase}%
\algtext*{EndSwitch}%
\algtext*{EndCase}%

\algnewcommand\algorithmicwithprob{\textbf{with probability}}
\algnewcommand\algorithmicotherwise{\textbf{otherwise}}

\algdef{SE}[WithProb]{WithProb}{EndWithProb}[1]{\algorithmicwithprob\ #1\ \algorithmicdo}{\algorithmicend\ \algorithmicwithprob}%
\algdef{Ce}[Otherwise]{WithProb}{Otherwise}{EndWithProb}
  {\algorithmicotherwise}%
\algtext*{EndWithProb}%
\algtext*{EndOtherwise}%

\newcommand{\Prob}[1]{\text{Pr}\left[#1\right]\xspace}

\makeatletter
\newtheorem*{rep@theorem}{\rep@title}
\newcommand{\newreptheorem}[2]{%
\newenvironment{rep#1}[1]{%
 \def\rep@title{#2 \ref{##1}}%
 \begin{rep@theorem}}%
 {\end{rep@theorem}}}
\makeatother
\newreptheorem{lemma}{Lemma}
\newreptheorem{theorem}{Theorem}

\title{Fault-Tolerant Consensus with an Abstract MAC Layer\thanks{Peter Robinson acknowledges the support of the Natural Sciences and Engineering Research Council of Canada (NSERC), application ID RGPIN-2018-06322.
Calvin Newport acknowledges the support of the National Science Foundation, award number 1733842.}}
\author{Calvin Newport\\ Georgetown University\\ \texttt{cnewport@cs.georgetown.edu}
\and Peter Robinson\\ McMaster University\\ \texttt{peter.robinson@mcmaster.ca}}
\date{}

\begin{document}
\maketitle

\begin{abstract}
In this paper, we study fault-tolerant distributed consensus in wireless systems. 
In more detail, we produce two new randomized algorithms that solve this problem in the abstract MAC layer model,
which captures the basic interface and communication guarantees provided by most wireless MAC layers. 
Our algorithms work for any number of failures,  require no advance knowledge of the network participants or network size,
and guarantee termination with high probability after a number of broadcasts that are polynomial in the network size.
Our first algorithm satisfies the standard agreement property, while our second trades a faster termination guarantee in exchange
for a looser agreement property in which most nodes agree on the same value.
These are the first known fault-tolerant consensus algorithms for this model.
In addition to our main upper bound results,
we explore the gap between the abstract MAC layer and the standard asynchronous message passing model
by proving fault-tolerant consensus is impossible in the latter in the absence of information regarding the network participants,
even if we assume no faults, allow randomized solutions, and provide the algorithm a constant-factor approximation of the network size.
\end{abstract}

\section{Introduction}
\label{sec:intro}

Consensus provides a fundamental building block for developing reliable distributed
systems~\cite{guerraoui:1997,guerraoui:2000,guerraoui:2001}.
Accordingly, it is well studied in many different system models~\cite{lynch:1996}.
Until recently, however, little was known about solving this problem in distributed
systems made up of devices communicating using commodity wireless cards.
Motivated by this knowledge gap, 
this paper studies consensus in the {\em abstract MAC layer} model, 
which abstracts the basic behavior and guarantees of standard wireless MAC layers.
In recent work~\cite{newport:2014}, 
we proved deterministic fault-tolerant consensus is impossible in this setting.
In this paper, we describe and analyze the first known randomized
fault-tolerant consensus algorithms for this well-motivated model.

\smallskip
\noindent {\bf The Abstract MAC Layer.}
Most existing work on distributed algorithms for wireless networks assumes low-level synchronous models 
that force algorithms to directly grapple with issues caused by contention and signal fading.
Some of these models describe the network topology with a graph
(c.f.,~\cite{baryehuda:1987,jurdzinski:2002,kowalski:2005,moscibroda:2005,czumaj:2006,gasieniec:2007}),
while others use signal strength calculations to determine message behavior
(c.f.,~\cite{moscibroda:2006,moscibroda:2007,goussevskaia:2009,halldorsson:2012b,jurdzinski:2013:random,daum:2013}). 

As also emphasized in~\cite{newport:2014}, these models are useful for asking foundational questions about distributed computation
on shared channels, but  are not so useful for developing algorithmic strategies suitable for deployment.
In real systems, algorithms typically do not operate in synchronous rounds and they are not provided unmediated access to the radio.
They must instead operate on top of a general-purpose MAC layer 
which is responsible for many network functions, including contention management, rate control, and co-existence with other network traffic.

Motivated by this reality, in this paper we adopt the {\em abstract MAC layer} model~\cite{kuhn:2011abstract},
an asynchronous broadcast-based communication model
that captures the basic interfaces and guarantees provided by common existing wireless MAC 
layers. In more detail, if you provide the abstract MAC layer a message to broadcast,
it will eventually be delivered to nearby nodes in the network.
The specific means by which contention is managed---e.g., CSMA, TDMA, uniform probabilistic routines such as DECAY~\cite{baryehuda:1987}---is abstracted away
by the model. 
At some point after the contention management completes,
the abstract MAC layer passes back an {\em acknowledgment} indicating that it is ready for the next message.
This acknowledgment contains no information about the number or identities of the message recipient.

(In the case of the MAC layer using CSMA, for example, the acknowledgment would be generated after the MAC layer detects a clear channel.
In the case of TDMA, the acknowledgment would be generated after the device's turn in the TDMA schedule.
In the case of a probabilistic routine such as DECAY, the acknowledgment would be generated after a sufficient number of attempts to guarantee
successful delivery to all receivers with high probability.)

The abstract MAC abstraction, of course, does not attempt to provide a detailed representation of any specific existing MAC layer.
Real MAC layers offer many more modes and features then is captured by this model.
In addition, the variation studied in this paper assumes messages are always delivered, whereas more realistic variations would allow
for occasional losses.

This abstraction, however, still serves to 
capture the fundamental dynamics of real wireless application design in which the lower layers dealing directly
with the radio channel are separated from the higher layers executing the application in question.
An important goal in studying this abstract MAC layer, therefore, is attempting to uncover principles and strategies that can close the gap
between theory and practice in the design of distributed systems deployed on standard layered wireless architectures.

\smallskip
\noindent {\bf Our Results.}
In this paper, we studied randomized fault-tolerant consensus algorithms in the abstract MAC layer model.
In more detail, we study binary consensus and assume a single-hop network topology.
Notice, our use of randomization is necessary, as deterministic consensus is impossible in the abstract MAC layer model in the presence of even
a single fault (see our generalization of FLP from~\cite{newport:2014}).

To contextualize our results,
we note that the abstract MAC layer model differs from standard asynchronous message passing models
in two main ways: (1) the abstract MAC layer model provides the algorithm no advance information about the network size or membership,
requiring nodes to communicate with a blind broadcast primitive instead of using point-to-point channels,
(2) the abstract MAC layer model provides an acknowledgment to the broadcaster at some point after its message
has been delivered to all of its neighbors. This acknowledgment, however, contains no information about
the number or identity of these neighbors (see above for more discussion of this fundamental feature of standard wireless MAC layers).

Most randomized fault-tolerant consensus algorithms in the asynchronous message passing model strongly leverage
knowledge of the network. A strategy common to many of these algorithms, for example, is to repeatedly collect messages from
at least $n-f$ nodes in a network of size $n$ with at most $f$ crash failures (e.g.,~\cite{benor}).
This strategy does not work in the abstract MAC layer model as nodes do not know $n$.

To overcome this issue, we adapt an idea introduced in early work on fault-tolerant consensus in the asynchronous shared memory
model: {\em counter racing} (e.g.,~\cite{chandra,jim}).
At a high-level, this strategy has nodes with initial value $0$ advance a shared memory counter associated with $0$,
while nodes with initial value $1$ advance a counter associated with $1$.
If a node sees one counter get ahead of the other, they adopt the initial value associated with the larger counter,
and if a counter gets sufficiently far ahead, then nodes can decide.

Our first algorithm (presented in Section~\ref{slow}) implements a counter race of sorts using the acknowledged blind broadcast primitive provided by the model.
Roughly speaking, nodes continually broadcast their current proposal and counter, and update both based on the
pairs received from other nodes. 
Proving safety for this type of strategy in shared memory models is simplified by the atomic nature of register accesses.
In the abstract MAC layer model, by contrast, a broadcast message is delivered non-atomically to its recipients, 
and in the case of a crash, may not arrive at some recipients at all.\footnote{We note that register simulations are also not an option
in our model for two reasons: standard simulation algorithms require knowledge of $n$ and a majority correct nodes,
whereas we assume no knowledge of $n$ and wait-freedom.} 
Our safety analysis, therefore, requires novel analytical tools that tame a more diverse set of possible system configurations.

To achieve liveness, we use a technique loosely inspired by the randomized delay strategy introduced by Chandra in the shared memory model~\cite{chandra} .
In more detail, nodes probabilistically decide to replace certain sequences of their counter updates with $nop$ placeholders.
We show that if these probabilities are adapted appropriately, the system eventually arrives at a state where it becomes likely for only a single
node to be broadcasting updates, allowing progress toward termination.

Formally, we prove that with high probability in the network size $n$, the algorithm terminates after $O(n^3\log{n})$ broadcasts are scheduled.
This holds regardless of which broadcasts are scheduled (i.e., we do not impose a fairness condition), and regardless of the number
of faults. The algorithm, as described, assumes nodes are provided unique IDs that we treat as comparable black boxes (to prevent them from leaking network size information).
We subsequently show how to remove that assumption by describing an algorithm that generates unique IDs in this setting with high probability. 

Our second algorithm (presented in Section~\ref{fast}) trades a looser agreement guarantee for more efficiency.
 In more detail, we describe and analyze a solution to {\em almost-everywhere} agreement~\cite{dwork:1988},
 that guarantees most nodes agree on the same value.
 This algorithm terminates after $O(n^2\log^4{n}\log\log{n})$ broadcasts, which is a linear factor faster than our first algorithm (ignoring log factors).
 The almost-everywhere consensus algorithm consists of two phases.
 The first phase is used to ensure that almost all nodes obtain a good approximation of the network size. 
 In the second phase, nodes use this estimate to perform a sequence of broadcasts meant to help spread their proposal to the network.
 Nodes that did not obtain a good estimate in Phase~1 will leave Phase~2 early.
 The remaining nodes, however, can leverage their accurate network size estimates to probabilistically sample a subset to actively participate
 in each round of broadcasts.
 To break ties between simultaneously active nodes, each chooses a random rank using the estimate obtained in Phase~1. 
 We show that with high probability, after not too long,
 there exists a round of broadcasts in which the first node receiving its acknowledgment is both active and has the minimum rank among other active nodes---allowing its proposal to spread to all remaining nodes.
 
Finally, we explore the gap between the abstract MAC layer model and the related asynchronous message passage passing model.
We prove (in Section~\ref{sec:lower}) that fault-tolerant consensus is impossible in the asynchronous message passing model
in the absence of knowledge of network participants, even if we assume no faults, allow randomized algorithms, and provide a constant-factor approximation of $n$.
This differs from the abstract MAC layer model where we solve this problem without network participant or network size information,
and assuming crash failures.
This result implies that the fact that broadcasts are acknowledged in the abstract MAC layer model is crucial to overcoming 
the difficulties induced by limited network information.

\smallskip
\noindent {\bf Related Work.}
Consensus provides a fundamental building block for reliable distributed computing~\cite{guerraoui:1997,guerraoui:2000,guerraoui:2001}.
It is particularly well-studied in asynchronous models~\cite{paxos,schiper:1997,mostefaoui:1999,aguilera:2000}.

The abstract MAC layer approach\footnote{There is no {\em one} abstract MAC layer model.
Different studies use different variations. They all share, however, the same general commitment to capturing the types of interfaces
and communication/timing guarantees that are provided by standard wireless MAC layers} 
to modeling wireless networks was introduced in~\cite{kuhn:2009}
(later expanded to a journal version~\cite{kuhn:2011abstract}), and
has been subsequently used to study several
 different problems~\cite{cornejo2009neighbor,khabbazian:2010,khabbazian:2011,cornejo2014reliable,newport:2014}. 
 The most relevant of this related work is~\cite{newport:2014},
 which was the first paper to study consensus in the abstract MAC layer model.
 This previous paper generalized the seminal FLP~\cite{flp} result
 to prove deterministic consensus is impossible in this model even in the presence of a single failure.
 It then goes on to study deterministic consensus in the absence of failures,
identifying the pursuit of fault-tolerant  {\em randomized} solutions as important future work---the challenge
 taken up here.

We note that other researchers have also studied consensus using high-level wireless network abstractions. 
Vollset and Ezhilchelvan~\cite{vollset:2005}, 
and Alekeish and Ezhilchelvan~\cite{alekeish:2012}, study consensus
in a variant of the asynchronous message passing model where pairwise channels come and go dynamically---capturing 
some behavior of {mobile} wireless networks. Their correctness results depend on detailed liveness guarantees
that bound the allowable channel changes.
Wu et~al.~\cite{wu:2009} use the standard asynchronous message passing model (with 
unreliable failure detectors~\cite{chandra:1996}) as a stand-in for a wireless network,
focusing on how to reduce message complexity (an important metric in a resource-bounded wireless setting)
in solving consensus.

A key difficulty for solving consensus in the abstract MAC layer model is the absence of advance information
about network participants or size.
These constraints have also been studied in other models.
Ruppert~\cite{ruppert2007anonymous},
and Bonnet and Raynal~\cite{bonnet2010anonymous},
for example, study the amount of extra power needed (in terms of shared objects and failure detection, respectively)
to solve wait-free consensus in {\em anonymous} versions of the standard models.
Attiya et~al.~\cite{attiya2002computing} describe consensus solutions for shared memory systems without failures or unique ids.
A series of papers~\cite{cavin:2004,greve:2007,alchieri:2008}, starting with the work of Cavin et~al.~\cite{cavin:2004},
study the related problem of  {\em consensus with unknown participants} (CUPs),
where nodes are only allowed to communicate with other nodes whose identities have been provided
by a {\em participant detector} formalism. 

Closer to our own model is the work of Abboud~et~al.~\cite{abboud:2008},
which also studies a single hop network where nodes broadcast messages to an unknown group of network participants.
They prove deterministic consensus is impossible
in these networks under these assumptions without knowledge of network size.
In this paper, we extend these existing results by proving this impossibility still holds even if we assume
randomized algorithms and provided the algorithm a constant-factor approximation of the network size. 
This bound opens a sizable gap with our abstract MAC layer model in which consensus is solvable without
this network information.

We also consider almost-everywhere (a.e.) agreement \cite{dwork:1988}, a weaker variant of consensus, where a small number of nodes are allowed to decide on conflicting values, as long as a sufficiently large majority agrees. 
Recently, a.e.\ agreement has been studied in the context of peer-to-peer networks (c.f.\ \cite{king:2006,augustine:2015}), where the adversary can isolate small parts of the network thus rendering (everywhere) consensus impossible. 
We are not aware of any prior work on a.e.\ agreement in the wireless settings.

\section{Model and Problem}
\label{sec:model}

In this paper, we study a variation of the {\em abstract MAC layer} model,
which describes system consisting of a single hop network of $n\geq 1$ computational devices (called {\em nodes} in the following) that communicate
wirelessly using communication interfaces and guarantees inspired by commodity wireless MAC layers.

In this model, nodes communicate with a $bcast$ primitive 
 that guarantees to eventually deliver the broadcast message to all the other nodes (i.e., the network is single hop).
At some point after a given $bcast$ has succeeded in delivering a message to all other nodes, the broadcaster
receives an $ack$ informing it that the broadcast is complete (as detailed in the introduction, this captures the reality
that most wireless contention management schemes have a definitive point at which they know a message broadcast is complete). 
This acknowledgment contains no information about
the number or identity of the receivers.

We assume a node can only broadcast one message at a time. That is, once it invokes $bcast$, it cannot broadcast another message until 
receiving the corresponding $ack$ (formally, overlapping messages are discarded by the MAC layer).
We also assume any number of nodes can permanently stop executing due to crash failures.
As in the classical message passing models, a crash can occur during a broadcast, meaning that some nodes might receive the message while others do not.

This model is event-driven with the relevant events scheduled asynchronously by an arbitrary {\em scheduler}.
In more detail,
for each node $u$, there are four event types relevant to $u$ that can be scheduled:  $init_u$ (which occurs at the beginning of an execution and allows $u$ to initialize), $recv(m)_u$ (which indicates that $u$ has received message $m$ broadcast from another node),
 $ack(m)_u$ (which indicates that the message $m$ broadcast by $u$ has been successfully delivered), and $crash_u$ (which indicates that $u$ is crashed for the remainder of the execution).

A distributed algorithm specifies for each node $u$
a finite collection of steps to execute for each of the non-$crash$ event types.
When one of these events is scheduled by the scheduler,
we assume the corresponding steps are executed atomically at the point that the event is scheduled.
Notice that one of the steps that a node $u$ can take in response to these events is to invoke a $bcast(m)_u$ primitive for some message $m$.
When an event includes a $bcast$ primitive we say it is {\em combined} with a broadcast.\footnote{Notice, we can assume without loss of generality, 
that the steps executed in response to an event never invoke more than a single $bcast$ primitive, as any additional broadcasts invoked at the same time
would lead to the messages being discarded due to the model constraint that a node must receive an $ack$ for the current message before broadcasting a new message.} 

We place the following constraints on the scheduler.
 It must start each execution by scheduling an $init$ event for each node; i.e., 
we study the setting where all participating nodes are activated at the beginning of the execution.
If a node $u$ invokes a valid $bcast(m)_u$ primitive,
then for each $v\neq u$ that is not crashed when the broadcast primitive is invoked,
the scheduler must subsequently either schedule a single $recv(m)_v$ or $crash_v$ event at $v$.
At some point after these events are scheduled, it must then eventually schedule an $ack(m)_u$ event at $u$.
These are the only $recv$ and $ack$ events it schedules (i.e., it cannot create new messages from scratch or cause messages
to be received/acknowledged multiple times).
If the scheduler schedules a $crash_v$ event, it cannot subsequently schedule any future events for $u$.

We assume that in making each event scheduling decision,
the scheduler can use the schedule history as well as the algorithm definition,
but it does not know the nodes' private states (which includes the nodes' random bits).
When the scheduler schedules an event that triggers a broadcast (making it a combined event), 
it is provided this information so that it knows it must now schedule receive events for the message.
We assume, however, that the scheduler does not learn the {\em contents} of the broadcast message.\footnote{This adversary model is sometimes called {\em message oblivious} and it is commonly
considered a good fit for schedulers that control network behavior. 
This follows because it allows the scheduler to adapt the schedule based on the number of messages being sent and their
sources---enabling it to model contention and load factors. One the other hand, there is not good justification for the idea that this schedule should somehow also depend on the specific bits
contained in the messages sent. Notice, our liveness proof specifically leverages the message oblivious assumption as it prevents the scheduler from knowing which nodes are sending updates
and which are sending $nop$ messages. }

Given an execution $\alpha$,
we say the {\em message schedule} for $\alpha$,
also indicated $msg[\alpha]$,
is the sequence of message events (i.e., $recv$, $ack$, and $crash$) scheduled in the execution.
We assume that a message schedule includes indications of which events are combined with broadcasts.

\smallskip
\noindent {\bf The Consensus Problem.} 
 In this paper, we study binary consensus with probabilistic termination.
 In more detail, at the beginning of an execution each node is provided an {\em initial value} from $\{0,1\}$ as input.
Each node has the ability to perform a single irrevocable $decide$ action for either value $0$ or $1$.
To solve consensus, an algorithm must guarantee the following three properties: (1) {\em agreement}: no two nodes
decide different values; (2) {\em validity}: if a node decides value $b$, then at least one node started with initial value $b$;
and (3) {\em termination (probabilistic)}: every non-crashed node decides with probability $1$ in the limit.

Studying finite termination bounds is complicated in asynchronous models because the scheduler can delay specific nodes taking
steps for arbitrarily long times.
In this paper, we circumvent this issue by proving bounds on the number of scheduled events before the system reaches
a {\em termination state} in which every non-crashed node has: (a) decided; or (b) will decide whenever the scheduler gets around to scheduling its next $ack$ event.

Finally, in addition to studying consensus with standard agreement, we also study {\em almost-everywhere} agreement, 
in which only a specified majority fraction (typically a $1-o(n)$ fraction of the $n$ total nodes) must agree.

\begin{algorithm}

\begin{algorithmic}
\State
\State \underline{Initialization:}  
\State $c_u \gets 0$
\State $n_u \gets 2$  %
\State $C_u \gets \{  (id_u, c_u,v_u)  \}$
\State $peers \gets \{ id_u\}$
\State $phase \gets 0$
\State $active \gets true$
\State $decide \gets -1$
\State $k \gets 3$ 
\State $c\gets k+3$
\State {\bf bcast}$(nop,id_u,n_u)$
\State

\State \underline{On Receiving $ack(m)$:}
\State $phase \gets phase +1$
\If{$m=(decide,b)$}
	 	\State {\bf decide}$(b)$ and {\bf halt}$()$
\Else
	\State $newm \gets \bot$
	\State $C_u' \gets C_u$
	\State $\hat c_u^{(0)} \gets$ max counter in $C_u'$ paired with value $0$ (default to $0$ if no such elements)
	\State $\hat c_u^{(1)} \gets$ max counter in $C_u'$ paired with value $1$ (default to $0$ if no such elements)
	
	\If{$\hat c_u^{(0)} > \hat c_u^{(1)}$} $v_u \gets 0$
	\ElsIf{$\hat c_u^{(1)} > \hat c_u^{(0)}$}  $v_u \gets 1$
	\EndIf
	
	\If{$\hat c_u^{(0)} \geq \hat c_u^{(1)} + k$ {\bf or} $decide = 0$} $newm \gets (decide,0)$
	\ElsIf{$\hat c_u^{(1)} \geq \hat c_u^{(0)} + k$ {\bf or} $decide =1$} $newm \gets (decide,1)$
	\EndIf
	
	\If{$newm = \bot$}
		\If{$\max\{\hat c_u^{(0)}, \hat c_u^{(1)} \} \leq c_u$ {\bf and} $m \neq nop$} $c_u \gets c_u + 1$
		\ElsIf{$\max\{\hat c_u^{(0)}, \hat c_u^{(1)} \} > c_u$} $c_u \gets max\{\hat c_u^{(0)}, \hat c_u^{(1)} \}$
		\EndIf
		
		\State {\bf update} $(id_u,*,*)$ element in $C_u$ with new $c_u$ and $v_u$
		\State $newm \gets (counter,id_u,c_u,v_u,n_u)$
	\EndIf
	
	\If{$phase\ \%\ c = 1$} with probability $1/n_u$ $active\gets true$ otherwise $active\gets false$ \EndIf
	\If{$newm = (decide,*)$ {\bf or} $active = true$} 
		\State {\bf bcast}$(newm)$
	\Else 
		\State {\bf bcast}$(nop,id_u,n_u)$
	\EndIf
\EndIf
\State

\State \underline{On Receiving Message $m$:}
\State {\bf updateEstimate}$(m)$
 \If{$m=(decide,b)$}
 	\State $decide \gets b$
 \ElsIf{$m=(counter,id,c,v,n')$}
 	\If{$\exists c',v'$ such that $(id,c',v') \in C_u$}
		\State {\bf remove} $(id,c',v')$ from $C_u$
	\EndIf
	 \State {\bf add} $(id,c,v)$ to $C_u$
 \EndIf

\end{algorithmic}

\label{alg:1}
\caption{Counter Race Consensus (for node $u$ with UID $id_u$ and initial value $v_u$)}
\end{algorithm}

\begin{algorithm}
\begin{algorithmic}
\If{$m$ contains a UID $id$ and network size estimate $n'$}
	\State $peers \gets peers \cup \{id\}$
	\State $n_u \gets \max\{ n_u,|peers|, n'\}$
\EndIf

\end{algorithmic}
\caption{The {\bf updateEstimate}$(m)$ subroutine called by Counter Race Consensus during $recv(m)$ event.}
\end{algorithm}

\section{Consensus Algorithm}
\label{slow}

Here we describe analyze our randomized binary consensus algorithm: {\em counter race consensus} (see Algorithms $1$ and $2$ for pseudocode,
and Section~\ref{sec:slow:alg} for a high-level description of its behavior).
This algorithm assumes no advance knowledge of the network participants or network size.
Nodes are provided unique IDs,
but these are treated as comparable black boxes, preventing them from leaking information about the network size.
(We will later discuss how to remove the unique ID assumption.)
It tolerates any number of crash faults.
\onlyShort{The detailed proofs can be found in the full paper \cite{fullpaper}.}

\subsection{Algorithm Description}
\label{sec:slow:alg}

The counter race consensus algorithm is described in pseudocode in the figures labeled Algorithm $1$ and $2$.
Here we summarize the behavior formalized by this pseudocode.

The core idea of this algorithm is that each node $u$ maintains a counter $c_u$ (initialized to $0$) and a proposal $v_u$ (initialized to its consensus initial value).
Node $u$ repeatedly broadcasts $c_u$ and $v_u$,  updating these values before each broadcast.
That is, during the $ack$ event for its last broadcast of $c_u$ and $v_u$,
node $u$ will apply a set of {\em update rules}  to these values. It then concludes the $ack$ event by broadcasting these updated values.
This pattern repeats until $u$ arrives at a state where it can safely commit to deciding a value.

The update rules and decision criteria applied during the $ack$ event are straightforward. 
Each node $u$ first calculates $\hat c_u^{(0)}$, the largest counter value it has sent or received in a message containing proposal value $0$,
and $\hat c_u^{(1)}$, the largest counter value it has sent or received in a message containing proposal value $1$.

If $\hat c_u^{(0)} > \hat c_u^{(1)}$, then $u$ sets $v_u \gets 0$, 
and if $\hat c_u^{(1)} > \hat c_u^{(0)}$, then $u$ sets $v_u \gets 1$.
That is, $u$ adopts the proposal that is currently ``winning" the counter race (in case of a tie, it does not change its proposal).

Node $u$ then checks to see if either value is winning by a large enough margin to support a decision.
In more detail, if $\hat c_u^{(0)} \geq \hat c_u^{(1)} + 3$,
then $u$ commits to deciding $0$, and if $\hat c_u^{(1)} \geq \hat c_u^{(0)} + 3$,
then $u$ commits to deciding $1$.

What happens next depends on whether or not $u$ committed to a decision. 
If $u$ did {\em not} commit to a decision (captured in the {\bf if} $newm = \bot$ {\bf then} conditional),
then it must update its counter value. 
To do so, it compares its current counter $c_u$ to $\hat c_u^{(0)}$ and $\hat c_u^{(1)}$.
If $c_u$ is smaller than one of these counters, it sets $c_u \gets \max\{ \hat c_u^{(0)}, \hat c_u^{(1)}\}$.
Otherwise, if $c_u$ is the largest counter that $u$ has sent or received so far, it
will set $c_u \gets c_u + 1$. Either way, its counter increases.
At this point, $u$ can complete the $ack$ event by broadcasting a message containing its newly updated
$c_u$ and $v_u$ values.

On the other hand, if $u$ committed to deciding value $b$, then it will send a $(decide,b)$ message to inform
the other nodes of its decision.
On subsequently receiving an $ack$ for this message, $u$ will decide $b$ and halt. 
Similarly, if $u$ ever receives a $(decide, b)$ message from {\em another} node,
it will commit to deciding $b$.
During its next $ack$ event, it will send its own $(decide,b)$ message and decide and halt on its corresponding $ack$.
That is, node $u$ will not decide a value until it has broadcast its commitment to do so, and received an $ack$ on the broadcast.

The behavior described above guarantees agreement and validity.
It is not sufficient, however, to achieve liveness, as an ill-tempered scheduler can conspire to keep the race between $0$ and $1$ too close for a decision commitment.
To overcome this issue we introduce a random delay strategy that has nodes randomly step away from
the race for a while by replacing their broadcast values with $nop$ placeholders ignored by those who receive them.
Because our adversary does not learn the {\em content} of broadcast messages,
it does not know which nodes are actively participating and which nodes are taking a break (as in both cases, nodes continually broadcast messages)---thwarting
its ability to effectively manipulate the race.

In more detail, each node $u$ partitions its broadcasts into {\em groups} of size $6$.
At the beginning of each such group, $u$ flips a weighted coin to determine whether or not to replace the counter and proposal
values it broadcasts during this group with $nop$ placeholders---eliminating its ability to affect other nodes' counter/proposal values.
As we will later elaborate in the liveness analysis,
the goal is to identify a point in the execution in which a single node $v$ is broadcasting its values while all other nodes are broadcasting $nop$ values---allowing
$v$ to advance its proposal sufficiently far ahead to win the race.

To be more specific about the probabilities used in this logic, node $u$ maintains an estimate $n_u$ of the number of nodes in the network.
It replaces values with $nop$ placeholders in a given group with probability $1/n_u$.
(In the pseudocode, the  $active$ flag indicates whether or not $u$ is using $nop$ placeholders in the current group.)
Node $u$ initializes $n_u$ to $2$. It then updates it by calling the {\em updateEstimate} routine (described in Algorithm $2$)
for each message it receives.

There are two ways for this routine to update $n_u$.
The first is if the number of unique IDs that $u$ has received so far (stored in $peers$) is larger than $n_u$. 
In this case, it sets $n_u \gets |peers|$.
The second way is if it learns another node has an estimate $n' > n_u$.
In this case, it sets $n_u \gets n'$.
Node $u$ learns about other nodes' estimates, as the algorithm has each node append its current estimate to all
of its messages (with the exception of $decide$ messages).
In essence, the nodes are running a network size estimation routine parallel to its main counter race logic---as nodes refine
their estimates, their probability of taking useful breaks improves.

\subsection{Safety}
We begin our analysis by proving
that our algorithm satisfies the agreement and validity properties of the consensus problem.
Validity follows directly from the algorithm description.
Our strategy to prove agreement is to show that if any node sees a value $b$ with a counter at least $3$ ahead of value $1-b$ (causing it to commit to deciding $b$),
then $b$ is the only possible decision value.
Race arguments of this type are easier to prove in a shared memory setting where nodes work with objects like atomic registers that guarantee linearization points.
In our message passing setting, by contrast,
in which broadcast messages arrive at different receivers at different times,
we will require  more involved definitions and operational arguments.\footnote{We had initially hoped there might be some way to simulate linearizable shared objects
in our model. Unfortunately, our nodes' lack of information about the network size thwarted standard simulation strategies which typically require nodes to collect messages
from a majority of nodes in the network before proceeding to the next step of the simulation.}

We start
with a useful definition.
We say $b$ {\em dominates} $1-b$ at a given point in the execution,
if every (non-crashed) node at this point believes $b$ is winning the race, and none of the messages
in transit can change this perception.

To formalize this notion we need some notation.
In the following, we say {\em at point $t$} (or {\em at $t$}), with respect to an event $t$ from the message schedule of an execution $\alpha$,
to describe the state of the system immediately after event $t$ (and any associated steps that execute atomically with $t$) occurs.
We also use the notation {\em in transit at $t$} to describe messages that have been broadcast but not yet received at every non-crashed receiver
at $t$.

\begin{definition}
Fix an execution $\alpha$, event $t$ in the corresponding message schedule $msg[\alpha]$, consensus value $b\in \{0,1\}$, and counter value $c\geq 0$.
We say $\alpha$ is {\em $(b,c)$-dominated} at $t$ if the following conditions are true:
\begin{compactenum}
\item For every node $u$ that is not crashed at $t$: $\hat c_u^{(b)}[t] > c$ and $\hat c_u^{(1-b)}[t] \leq c$,
where at point $t$, $\hat c_u^{(b)}[t]$ (resp. $\hat c_u^{(1-b)}[t]$) is the largest value $u$ has sent or received in a counter message containing consensus value $b$ (resp. $1-b$).
If $u$ has not sent or received any counter messages containing $b$ (resp. $1-b$), then by default it sets $\hat c_u^{(b)}[t] \gets 0$ (resp. $\hat c_u^{(1-b)}[t] \gets 0$) in making this comparison.

\item For every message of the form $(counter,id,1-b,c',n')$ that is in transit at  $t$: $c' \leq c$.

\end{compactenum} 

\label{def:dom}
\end{definition}

The following lemma formalizes the intuition that once an execution becomes dominated by a given value, it remains dominated by this value.

\begin{lemma}
Assume some execution $\alpha$ is $(b,c)$-dominated at point $t$. It follows that $\alpha$ is $(b,c)$-dominated at every $t'$ that comes after $t$.
\label{lem:dom}
\end{lemma}
\begin{proof}
In this proof, we focus on the suffix of the message schedule $msg[\alpha]$ that begins with event $t$.
For simplicity, we label these events $E_1,E_2,E_3,...$, with $E_1 = t$.
We will prove the lemma by induction on this sequence.

The base case ($E_1$) follows directly from the lemma statement. 
For the inductive step, we must show that if $\alpha$ is $(b,c)$-dominated at point $E_{i}$, then it will be dominated at $E_{i+1}$ as well.
By the inductive hypothesis, we assume the execution is dominated immediately before $E_{i+1}$ occurs.
Therefore, the only way the step is violated is if $E_{i+1}$ transitions the system from dominated to non-dominated status.
We consider all possible cases for $E_{i+1}$ and show none of them can cause such a transition.

The first case is if $E_{i+1}$ is a $crash_u$ event for some node $u$. It is clear that a crash cannot transition a system into non-dominated status.

The second case is if $E_{i+1}$ is a $recv(m)_u$ event for some node $u$. This event can only transition the system into a non-dominated status
if $m$ is a counter message that includes $1-b$ and a counter $c' > c$.
For $u$ to receive this message, however, means that the message was in transit immediately before $E_{i+1}$ occurs.
Because we assume the system is dominated at $E_i$, 
however, no such message can be in transit at this point (by condition $2$ of the domination definition).

The third and final case is if $E_{i+1}$ is a $ack(m)_u$ event for some node $u$, that is combined with a $bcast(m')_u$
event, where $m'$ is a counter message that includes $1-b$ and a counter $c' > c$.
Consider the values $\hat c_u^{(b)}$ and $\hat c_u^{(1-b)}$ set
by node $u$ early in the steps associated with this $ack(m)_u$ event.
By our inductive hypothesis, 
which tells us that the execution is dominated right before this $ack(m)_u$ event occurs,
it must follow that $\hat c_u^{(b)} > \hat c_u^{(1-b)}$ (as $\hat c_u^{(b)} = \hat c_u^{(b)}[E_{i}]$
and $\hat c_u^{(1-b)} = \hat c_u^{(1-b)}[E_{i}]$).
In the steps that immediately follow, therefore,
node $u$ will set $v_u \gets b$.
It is therefore impossible for $u$ to then broadcast a counter message with value $v_u = 1-b$.
\end{proof}

To prove agreement, we are left to show that if a node commits to deciding some value $b$, then it must be the case
that $b$ dominates the execution at this point---making it the only possible decision going forward.
The following helper lemma, 
which captures a useful property about counters,
will prove crucial for establishing this point.

\begin{lemma}
Assume event $t$ in the message schedule of execution $\alpha$ is combined with a $bcast(m)_v$, where $m=(counter, id_v, c,b,n_v)$,
for some counter $c>0$. It follows that prior to $t$ in $\alpha$, every node that is non-crashed at $t$ received a counter message
with counter $c-1$ and value $b$.
\label{lem:inc}
\end{lemma}
\begin{proof}
Fix some $t$, $\alpha$, $v$ and $m=(counter, id_v, c,b,n_v)$, as specified by the lemma statement.
Let $t'$ be the first event in $\alpha$ such that at $t'$ some node $w$ has local counter $c_w \geq c$ and value $v_w = b$.
We know at least one such event exists as $t$ and $v$ satisfy the above conditions, so the earliest such event, $t'$, is well-defined.
Furthermore, because $t'$ must modify local counter and/or consensus values, it must also be an $ack$ event.

For the purposes of this argument, let $c_w$ and $v_w$ be $w$'s counter and consensus value, respectively, immediately before $t'$ is scheduled.
Similarly, let $c_w'$ and $v_w'$ be these values immediately after $t'$ and its steps complete (i.e., these values at point $t'$).
By assumption: $c_w' \geq c$ and $v_w'=b$.
We proceed by studying the possibilities for $c_w$ and $v_w$ and their relationships with $c_w'$ and $v_w'$.

We begin by considering $v_w$.
We want to argue that $v_w=b$.
To see why this is true, assume for contradiction that $v_w=1-b$.
It follows that early in the steps for $t'$, node $w$ switches its consensus value from $1-b$ to $b$.
By the definition of the algorithm,
it only does this if at this point in the $ack$ steps:
 $\hat c_w^{(b)} > \hat c_w^{(1-b)} \geq c_w$ (the last term follows because $c_w$ is included in the values considered when defining $c_w^{(1-b)}$).
Note, however, that $c_w^{(b)}$ must be less than $c$. If it was greater than or equal to $c$,  
this would imply that a node ended an earlier event with counter $\geq c$ and value $b$---contradicting our assumption that $t'$ was the earliest such event.
If $c_w^{(b)} < c$ and $c_w^{(b)} > c_w$,
then $w$ must increase its $c_w$ value during this event.
But because  $\hat c_w^{(b)} > \hat c_w^{(1-b)}\geq c_w$, the only allowable change to $c_w$ would be to set it to $\hat c_w^{(b)} < c$.
This contradicts the assumption that $c_w' \geq c$.

At this checkpoint in our argument we have argued that $v_w=b$.
We now consider $c_w$.
If $c_w \geq c$, then $w$ starts $t'$ with a sufficiently big counter---contradicting the assumption
that $t'$ is the earliest such event.
It follows that $c_w < c$ and $w$ must increase this value during this event.

There are two ways to increase a counter; i.e., the two conditions in the {\em if/else-if} statement that follows the $newm = \bot$ check.
We start with the second condition.
If $\max\{\hat c_w^{(b)}, \hat c_w^{(1-b)}\} > c_w$, then $w$ can set $c_w$ to this maximum.
If this maximum is equal to $\hat c_w^{(b)}$, then this would imply $\hat c_w^{(b)} \geq c$.
As argued above, however, it would then follow that a node had a counter $\geq c$ and value $b$ before $t'$.
If this is not true, then $\hat c_w^{(1-b)} > c_w^{(b)}$.
If this was the case, however, $w$ would have adopted value $1-b$ earlier in the event,
contradicting the assumption that $v_w' = b$.

At this next checkpoint in our argument we have argued that $v_w =b$, $c_w < c$,
and $w$ increases $c_w$ to $c$ through the first condition of the {\em if/else if};
i.e., it must find that $\max\{\hat c_w^{(b)}, \hat c_w^{(1-b)}\} \leq c_w$ and $m\neq nop$.
Because this condition only increases the counter by $1$, we can further refine our assumption to $c_w = c-1$.

To conclude our argument, consider the implications of the  $m\neq nop$ component of this condition.
It follows that $t'$ is an $ack(m)_w$ for an actual message $m$.
It cannot be the case that $m$ is a $decide$ message,
as $w$ will not increase its counter on acknowledging a $decide$.
Therefore, $m$ is a counter message.
Furthermore, because counter and consensus values are not modified after broadcasting a counter message but before receiving its subsequent acknowledgment,
we know $m=(counter, id_w, c_w,v_w,*) = (counter, id_w, c-1, b,*)$ (we replace the network size estimate with a wildcard here as these estimates could
change during this period).

Because $w$ has an acknowledgment for this $m$,
by the definition of the model, prior to $t'$: every non-crashed node received a counter message with counter $c-1$ and consensus value $b$.
This is exactly the claim we are trying to prove.
\end{proof}

Our main safety theorem leverages the above two lemmas to establish that committing to decide $b$ means that $b$ dominates the execution. The key idea is that counter values cannot become too stale.
By Lemma~\ref{lem:inc}, if some node has a counter $c$ associated with proposal value $1-b$, then
all nodes have seen a counter of size at least $c-1$ associated with $1-b$.
It follows that if {some} node thinks $b$ is far ahead, then all nodes must think $b$ is far ahead in the race (i.e.,
$b$ dominates). Lemma~\ref{lem:dom} then establishes that this dominance is permanent---making $b$ the only
possible decision value going forward.

\begin{theorem}
The Counter Race Consensus algorithm satisfies validity and agreement.
\label{safety}
\end{theorem}
\begin{proof}
Validity follows directly from the definition of the algorithm.
To establish agreement, fix some execution $\alpha$ that includes at least one decision.
Let $t$ be the first $ack$ event in $\alpha$ that is combined with a broadcast of a $decide$ message.
We call such a step a {\em pre-decision} step as it prepares nodes to decide in a later step.
Let $u$ be the node at which this $ack$ occurs and $b$ be the value it includes in the $decide$ message.
Because we assume at least one process decides in $\alpha$, we know $t$ exists.
We also know it occurs before any decision. 

During the steps associated with $t$, 
$u$ sets $newm\gets (decide,b)$.
This indicates the following is true: $\hat c_u^{(b)} \geq \hat c_u^{(1-b)} + 3.$
Based on this condition, we establish two claims about the system at $t$, expressed with respect to the value $\hat c_u^{(1-b)}$ during these steps:

\begin{itemize}

 \item {\em Claim 1.} The largest counter included with value $1-b$ in a counter message broadcast\footnote{Notice, in these claims,
 when we say a message is ``broadcast" we only mean that the corresponding $bcast$ event occurred. We make no assumption on which nodes
 have so far received this message.} before $t$ is no more than $\hat c_u^{(1-b)} + 1$.
 
 Assume for contradiction that before $t$ some $v$ broadcast a counter message with value $1-b$ and counter $c > \hat c_u^{(1-b)} + 1$.
 By Lemma~\ref{lem:inc}, it follows that before $t$ every non-crashed node receives a counter message with value $1-b$ and counter $c-1 \geq \hat c_u^{(1-b)} + 1$.
 This set of nodes includes $u$. This contradicts our assumption that at $t$ the largest counter $u$ has seen associated with $1-b$ is $\hat c_u^{(1-b)}$.
 
 \item {\em Claim 2.} Before $t$, every non-crashed node has sent or received a counter message with value $b$ and counter at least $\hat c_u^{(1-b)}+2$.
 
 By assumption on the values $u$ has seen at $t$, we know that before $t$ some node $v$ broadcast a counter message with value $b$ and counter $c \geq \hat c_u^{(1-b)}+3$.
 By Lemma~\ref{lem:inc}, it follows that before $t$, every node has sent or received a counter with value $b$ and counter $c-1 \geq \hat c_u^{(1-b)}+2$.
 
 \end{itemize}
 
Notice that claim 1 combined with claim 2 implies that the execution is $(b,\hat c_u^{(1-b)}+1)$-dominated before $t$.
By Lemma~\ref{lem:dom}, the execution will remain dominated from this point forward.  
We assume $t$ was the first pre-decision, and it will lead $u$ to tell other nodes to decide $u$ before doing so itself.
Other pre-decision steps might occur, however, before all nodes have received $u$'s preference for $b$.
With this in mind, let $t'$ be any other pre-decision step.
Because $t'$ comes after $t$ it will occur in a $(b,\hat c_u^{(1-b)}+1)$-dominated system.
This means that during the first steps of $t'$, the node will adopt $b$ as its value (if it has not already done so),
meaning it will also promote $b$.

To conclude, we have shown that once any node reaches a pre-decision step for a value $b$, 
then the system is already dominated in favor of $b$, and
therefore $b$ is the only possible decision value going forward. Agreement follows directly.
\end{proof}
\subsection{Liveness}

We now turn our attention liveness. Our goal is to prove the following theorem:
 
\begin{theorem}
With high probability, within $O(n^3\ln{n})$ scheduled $ack$ events, every node executing counter race consensus has either crashed, decided, or received a $decide$ message. In the limit, this termination condition occurs with probability $1$.
\label{live:thm:main}
\end{theorem}

Notice that this theorem does not require a fair schedule. It guarantees its termination criteria (with high probability)
after {\em any} $O(n^3\ln{n})$ scheduled $ack$ events, regardless of {\em which} nodes these events occur at.
Once the system arrives at a state in which every node has either crashed, decided, or received a $decide$ message,
the execution is now univalent (only one decision value is possible going forward), and each non-crashed node $u$ will decide after at most two additional $ack$ events at 
$u$.\footnote{In the case where $u$ receives a $decide$ message, the first $ack$ might correspond to the message
it was broadcasting when the $decide$ arrived, and the second $ack$ corresponds to the $decide$ message 
that $u$ itself will then broadcast. During this second $ack$, $u$ will decide and halt.}

Our liveness proof is longer and more involved than our safety proof. This follows, in part, from the need to
introduce multiple technical definitions to help identify the execution fragments sufficiently well-behaved for
us to apply our probabilistic arguments.
With this in mind, we divide the presentation of our liveness proof into two parts.
The first part introduces the main ideas of the analysis and provides a road map of sorts to its component pieces.
The second part, which contains the details, can be found in the full paper \cite{fullpaper}.

\subsubsection{Main Ideas}

Here we discuss the main ideas of our liveness proof.
A core definition used in our analysis is the notion of an {\em $x$-run}.
Roughly speaking, for a given constant integer $x \geq 2$ and node $u$, we say an execution fragment $\beta$ is an $x$-run for some node $u$, 
if it starts and ends with an $ack$ event for $u$, it contains $x$ total $ack$ events for $u$, and no other node has more than $x$ 
$ack$ events interleaved. We deploy a recursive counting argument to establish that an execution fragment $\beta$ that contains at least $n\cdot x$
total $ack$ events, must contain a sub-fragment $\beta'$ that is an $x$-run for some node $u$.

To put this result to use,  we focus our attention on $(2c+1)$-runs, where $c=6$ is the constant used in the algorithm definition
to define the length of a {\em group} (see Section~\ref{sec:slow:alg} for a reminder of what a group is and how it is used by the algorithm).
A straightforward argument establishes that a $(2c+1)$-run for some node $u$
must contain at least one {\em complete group} for $u$---that is, it must contain all $c$ broadcasts of one of $u$'s groups.

Combining these observations, it follows that if we partition an execution into {\em segments} of length $n\cdot(2c+1)$,
each such segment $i$ contains a $(2c+1)$-run for some node $u_i$,
and each such run contains a complete group for $u_i$.
We call this complete group the {\em target group} $t_i$ for segment $i$ (if there are multiple complete groups in the run, 
choose one arbitrarily to be the target).

These target groups are the core unit to which our subsequent analysis applies.
Our goal is to arrive at a target group $t_i$ that is {\em clean} in the sense that $u_i$ is $active$ during the group (i.e., sends its actual values instead 
of $nop$ placeholders), and all broadcasts that arrive at $u$ during this group come from {\em non-active} nodes (i.e., these received messages
contain $nop$ placeholders instead of values). If we achieve a {\em clean} group, then it is not
hard to show that $u_i$ will advance its counter at least $k$ ahead of all other counters, pushing all other nodes into the termination criteria
guaranteed by Theorem~\ref{live:thm:main}.

To prove clean groups are sufficiently likely, our analysis must overcome two issues.
The first issue concerns network size estimations.
Fix some target group $t_i$. 
Let $P_i$ be the nodes from which $u_i$ receives at least one message during $t_i$.
If all of these nodes have a network size estimate of at least $n_i = |P_i|$ at the start of $t_i$,
we say the group is {\em calibrated.}
We prove that if $t_i$ is calibrated, then it is clean with a probability in $\Omega(1/n)$.

The key, therefore, is  proving most target groups are calibrated.
To do so, we note that if some $t_i$ is not calibrated, it means at least one node used an estimate
strictly less than $n_i$ when it probabilistically defined $active$ at the beginning of this group. 
During this group, however, all nodes will receive broadcasts from at least
$n_i$ unique nodes, increasing all network estimates to size at least $n_i$.\footnote{This summary is eliding
some subtle details tackled in the full analysis concerning which broadcasts are guaranteed to be received during a target group.
But these details are not important for understanding the main logic of this argument.}
Therefore, each target group that fails to be calibrated increases the minimum network size estimate in the system by
at least $1$. It follows that at most $n$ target groups can be non-calibrated.

The second issue concerns probabilistic dependencies. 
Let $E_i$ be the event that target group $t_i$ is clean and $E_j$ be the event that some other target group $t_j$ is clean.
Notice that $E_i$ and $E_j$ are not necessarily independent. If a node $u$ has a group that overlaps both $t_i$ and $t_j$,
then its probabilistic decision about whether or not to be active in this group impacts the potential cleanliness of both $t_i$ 
and $t_j$. 

Our analysis tackles these dependencies by identifying a subset of target groups that are pairwise independent. 
To do so, roughly speaking, we process our target groups in order. Starting with the first target group, we mark as
unavailable any future target group that overlaps this first group (in the sense described above). We then proceed until we arrive at the next target
group {\em not} marked unavailable and repeat the process.
Each available target group marks at most $O(n)$ future groups as unavailable. 
Therefore, given a sufficiently large set $T$ of target groups,
we can identify a subset $T'$, with a size in $\Omega(|T|/n)$, such that all groups in $T'$ are pairwise independent.

We can now pull together these pieces to arrive at our main liveness complexity claim.
Consider the first $O(n^3\ln{n})$ $ack$ events in an execution.
We can divide these into $O(n^2\ln{n})$ segments of length $(2c+1)n \in \Theta(n)$.
We now consider the target groups defined by these segments.
By our above argument, there is a subset $T'$ of these groups, where
$|T'| \in  \Omega(n\ln{n})$, and all target groups in $T'$ are mutually independent.
At most $n$ of these remaining  target groups are not calibrated. 
If we discard these, we are left with a slightly smaller set, of size still $\Omega(n\ln{n})$,
that contains only calibrated and pairwise independent target groups.

We argued that each calibrated group has a probability in $\Omega(1/n)$ of being clean.
Leveraging the independence between our identified groups, 
a standard concentration analysis establishes with high probability in $n$
that at least one of these $\Omega(n/\ln{n})$ groups  is clean---satisfying the Theorem statement.

\onlyLong{
\subsubsection{Full Analysis}

Our proof of Theorem~\ref{live:thm:main} proceeds in two steps. The first step introduces useful notation for describing parts of message schedules,
and proves some deterministic properties regarding these concepts.
The second step leverages these definitions and properties in making the core probabilistic arguments.

\paragraph{Definitions and Deterministic Properties}
Each node  keeps a counter called $phase$.
This counter is initialized to $0$ and is incremented with each $ack$ event. 
Given a message schedule and node $u$, we
can divide the schedule into {\em phases} with respect to $u$ based on $u$'s local $phase$ counter.
In more detail, label the $ack_u$ events in the schedule, $a_1,a_2,a_3...$.
For each $i \geq 1$,
we define {\em phase $i$} (with respect to $u$) to be the schedule fragment that starts with acknowledgment
$a_i$ and includes all events up to but {\em not} including $a_{i+1}$.
If no such $a_{i+1}$ exists (i.e., if $a_i$ is the last $ack_u$ event in the execution), we consider phase $i$ undefined and consider $u$ to only have $i-1$ phases in this schedule.
Notice, by our model definition, during a given phase $i$, all non-crashed nodes receive the message broadcast as part of the $ack$ that starts the phase.

We partition a given node $u$'s phases into {\em groups}, which we define with respect to the constant $c$ used in the algorithm
definition as part of the logic for resetting the nodes' $active$ flag.
In particular, we partition the phases into groups of size $c$. 
For a given node $u$,
phases $1$ to $c$ define group $1$, phases $c+1$ to $2c$ define group $2$, and, more generally,
for all $i\geq 1$, phases $(i-1)c + 1$ to $i\cdot c$ define group $i$.
Notice, by the definition of our algorithm, a node only updates its $active$ flag at the beginning of each group.
Therefore, the messages sent by a give node during a given one of its groups are either all $nop$ messages,
or all non-$nop$ messages.

We now introduce the higher level concept of a {\em run}, which will prove useful going forward.

\begin{definition}
Fix an execution $\alpha$ with corresponding message schedule $msg[\alpha]$, an integer $x\geq 2$, and a node $u$. 
We call a subsequence $\beta$ of $msg[\alpha]$ an {\em $x$-run for $u$} if it satisfies the following three properties:
\begin{compactenum}
 \item $\beta$ starts and ends with an $ack_u$ event,
 \item  $\beta$ contains $x$ total $acks$ for $u$, and 
 \item no other node has more than $x$ $acks$ in $\beta$.
\end{compactenum}
\label{live:def:x}
\end{definition}

We now show that for any $x$, any sufficiently long (defined with respect to $x$) fragment from a message schedule will contain an $x$-run for some node: 

\begin{lemma}
Fix an execution $\alpha$ and integer $x\geq 2$.
Let $\gamma$ be any subsequence of the corresponding message schedule $msg[\alpha]$ that includes
at least $n\cdot x$ $ack$ events.
There exists a subsequence $\beta$ of $\gamma$ that is an $x$-run for some node $u$.
\label{live:lem:x}
\end{lemma}
\begin{proof}
 Because $\gamma$ contains $n\cdot x$ total $acks$, 
 a straightforward counting argument provides that at least one node $v$ has at least $x$ $acks$ in $\gamma$.
 Consider the the subsequence $\gamma'$ of $\gamma$ that starts with the first $ack_v$ event and ends with the $x^{th}$ such $ack_v$ event.
 (That is, we remove the prefix of $\gamma$ before the first $ack_v$ and the suffix after the $x^{th}$ $ack_v$ event.)

It is clear that $\gamma'$ satisfies the first properties of our definition of an $x$-run for $v$. 
If it also satisfies the third property (that no {\em other} node has more than $x$ $acks$ in $\gamma'$), then we are done: setting $\beta \gets \gamma'$ satisfies the lemma statement.

On the other hand, if $\gamma'$ does not satisfy the third property, there
must exist some node $u$ that has more than $x$ $ack$ events in $\gamma'$. 
In this case, we can apply the above argument recursively to $u$ and $\gamma'$, 
identifying a subsequence of $\gamma'$ that starts with the first $ack_u$ and ends after the $x^{th}$ such event.
The resulting $\gamma''$ satisfies the first two properties of the definition of an $x$-run for $u$.
If it also satisfies the third property, we are done.
Otherwise, we can recurse again on $\gamma''$.

Because each such recursive application of this argument strictly reduces the size of the subsequence (at the very least, you are trimming off the first and last $ack$),
and the original $\gamma$ has a bounded number of events, the recursion must eventually arrive at a subsequence that satisfies all three properties of the $x$-run definition.
\end{proof}

We next prove an additional useful property of $x$-runs. In particular, a $(2c+1)$-run defined for some node $u$ is long enough that it must contain all phases of at least one
of $u$'s groups. Identifying {\em complete} groups of this type will be key to the later probabilistic algorithms.

\begin{lemma}
Let $\beta$ be a $(2c+1)$-run for some node $u$.
It follows that $\beta$ contains all of the phases for at least one of $u$'s groups (i.e., a {\em complete group} for $u$). 
\label{live:lem:target}
\end{lemma}
\begin{proof}
Because $x = 2c+1$, $\beta$ must contain at least $2c+1$ $ack_u$ events.
It follows that it contains 
at least $2c$ of node $u$'s phases (extra final $ack$ of the $2c+1$ ensures that all of the events that define 
phase $2c$ of the run are included in the run).
Because each node $u$ group consists of $c$ phases,
any sequence of $2c$ phases must include all $c$ phase of at least one full group.
\end{proof}

We next introduce the notion of a {\em clean} group,
and establish that the occurence of a clean group guarantees that we arrive at the termination state from our main theorem.

\begin{definition}
Let $\beta$ be a complete group for some node $u$.
We say $\beta$ is {\em clean} if the following two properties are satisfied:
\begin{enumerate}
\item Node $u$ sets $active$ to $true$ at the beginning of the group described by $\beta$. 
\item For every $recv_u(m)$ event that occurs in the first $c-1$ phases of $\beta$, $m$ is a $nop$ message.
(We do not restrict the messages received during the final phase of the clean group.)
\end{enumerate}
\label{live:def:clean}
\end{definition}

\begin{lemma}
Fix some execution $\alpha$. Assume fragment $\beta$ from $\alpha$ is a clean group for some node $u$.
It follows that by the end of $\beta$ all nodes have either crashed, decided, or received a $decide$ message.
\label{live:lem:decide}
\end{lemma}
\begin{proof}
Fix some $\alpha$, $\beta$ and $u$ as specified by the lemma statement.
Let $b$ be the consensus value $u$ adopts for the first phase of the clean group.
Because $u$ only receives $nop$ messages during all but the last phase of a clean group, we know $u$ will not change this value again in this group until (potentially)
the last phase. As we will now argue, however, it will have already decided before this last phase, so the fact that $u$ might receive values in that phase is inconsequential.

In more detail, let $\hat c_u^{(b)}$ and $\hat c_u^{(1-b)}$ be the largest counter values that $u$ has seen for $b$ and $1-b$, respectively, 
by the time it completes the $ack$ that begins the first phase.
Because we just assumed that $u$ adopts $b$ at this point, we know $\hat c_u^{(b)} \geq \hat c_u^{(1-b)}$.
Furthermore, because $u$ only receives $nop$ messages, we know that in every phase starting with phase $2$ of the group, $u$ will
either increment the counter associated $b$ or send a $decide$ message.
The largest counter associated with $1-b$ will not increase beyond $\hat c_u^{(1-b)}$ during these phases.

It follows that if $u$ has not yet sent a $decide$ message by the start of phase $k+2$, 
it will see during the $ack$ event that starts this phase that its largest counter for $b$ is $k$ larger than the largest counter for $1-b$.
Accordingly, during this phase $u$ will send a $decide$ message. During the $ack$ event that starts $k+3$,
$u$ will receive this $ack$ and decide. At this point, all other nodes have received its $decide$ message as well.
Because this is the last phase of the group, it is possible that $u$ receives non-$nop$ messages from other nodes---but at this point, this is too late to have an impact as $u$ has already
decided and halted.
(It is here that we see why $k+3$ is the right value for the group length $c$.)
\end{proof}

\paragraph{Randomized Analysis} 
In Part 1 of this analysis we introduced several useful definitions and execution properties.
These culminated with the argument in Lemma~\ref{live:lem:decide} that if we ever get a clean group in an execution,
then we will have achieved the desired termination property.
Our goal in this second part of the analysis is to leverage the tools from the preceding part 
to prove, with high probability, that the algorithm will generate a clean group after not too many $acks$ are scheduled.

\smallskip

\noindent {\em On Network Size Assumptions.}
If $n=1$, then all that is required for the single node $u$ to experience a clean group is for it to set $active$ to $true$.
By Lemma~\ref{live:lem:decide}, it will then decide and halt in the group that follows. 
By the definition of the algorithm, this occurs with probability $1/2$ at the beginning of each group,
as $u$ initializes $n_u \gets 2$, and this will never change.
Therefore: with high probability, $u$ will decide within $O(\log{n})$ groups (and therefore, $O(c\log{n})$ scheduled $acks$),
and with probability $1$, it will decide in the limit. This satisfies our liveness theorem. 
In the analysis that follows, therefore, we will assume $n>1$.

\smallskip

\noindent {\em On Independence Between the Schedule and Random Choices.}
According to our model assumptions (Section~\ref{sec:model}), the scheduler is provided no advance information
about the nodes' state or the contents of the messages they send. All the scheduler learns is the input assignment, 
and whether or not a given
node sent {\em some} message (but not the message contents) as part of the steps it takes for a given $init$ or $recv$ event. 
By the definition of our algorithm, however, until it halts, each node sends a message when initialized and after every $ack$,
regardless of its random choices or the specific contents of the messages its receives.
It follows that the scheduler learns nothing new about the nodes' states beyond their input values until the first node halts---at which point,
some additional information might be inferred.
For a node to halt, however, means it has already sent a $decide$ message and received an $ack$ for this message, meaning that we have
already satisfied the desired termination property at this point.
Accordingly, in the analysis that follows, we can treat the scheduler's choices as independent of the nodes' random choices.
This allows us to fix the schedule first and then reason probabilistically about the messages sent during the schedule,
without worrying about dependence between the schedule and those choices.\footnote{Technically speaking, in the analysis above, we imagine, without loss of generality,
that the scheduler
creates an infinite schedule that describes how it wants the execution to unfold {\em until} it learns the first node halts. At that point,
it can modify the schedule going forward. }

\smallskip

In analyzing the probability of a group ending up clean, a key property is whether or not the nodes participating in that group all have
good estimates of the network size (e.g., their $n_v$ values used in setting their $active$ flags).
We call a group with good estimates a {\em calibrated} group.
The formal definition of this property requires some care to ensure it exactly matches how we later study it:

\begin{definition}
Fix an execution $\alpha$.
Let $\beta$ be a complete group for some node $u$ in the message schedule $msg[\alpha]$.
Let $P_{\beta}$ be the set of nodes that have at least one of their messages received
by $u$ in the first $c-1$ phases of $u$'s group, let $n_{\beta} = |P_{\beta}|$, and
for each $v\in P_{\beta}$, let $t_v$ be the event in $msg[\alpha]$ that starts the node $v$ group that sends
the first of its messages received by $u$ in $\beta$.
We say that group $\beta$ is {\em calibrated} if for every $v\in P_{\beta}$:
the value $n_v$ used in event $t_v$ to probabilistically set $v$'s $active$ flag is of size at least $n_{\beta}$.
\label{live:def:cal}
\end{definition}

Notice in the above that if $P_{\beta}$ is empty than the property is vacuously true.
Another key property of calibration is that it is determined entirely by the message schedule. 
That is, given an prefix of a message schedule,
you can correctly determine the network size estimation of all nodes at the end of that prefix without needing to know anything about their input values or random choices.
This follows because network size estimates are based on two things: the number of UIDs from which you have received messages (of any type),
and other nodes' reported estimates (which are included on all message types). As argued above, the only thing impacted by the node random choices and inputs
are the types of messages they send, not {\em when} they send.

Therefore, given a message schedule and a group within the message schedule, we can determine whether or not that group is calibrated
independent of the nodes' random choices, supporting the following:

\begin{lemma}
Let $\alpha$ be a message schedule generated by the scheduler.
Let $\beta$ be a $(2c + 1)$-run for some node $u$ in $\alpha$, and $\gamma$ be a complete group
for $u$ in $\beta$.
If $\gamma$ is calibrated,
then the probability that $\gamma$ is clean is at least $1/(64n)$. 
\label{live:lem:prob}
\end{lemma}
\begin{proof}
Fix some $\alpha$, $\gamma$ and $\beta$ and $u$ as specified by the lemma statement. 
Fix $P_{\gamma}$, $n_{\gamma}$, and the $t_v$ events, as specified in our definition of calibrated (Definition~\ref{live:def:cal}).

We note that if $P_{\gamma}$ is empty, then the only condition that must hold for $\gamma$ to be clean is for $u$ to set $active$ to true.
This occurs with probability $1/n_u \geq 1/n > 1/(64n)$---satisfying the lemma.

Continuing, we consider the case where $P_{\gamma}$ is non-empty.
Fix any $v\in P_{\gamma}$. 
We begin by bounding the total number of $v$'s groups that might send a message that is received by $u$ in $\gamma$.
To do so, we note that because $\gamma$ is a $(2c+1)$-run,
$v$ cannot have more than $2c+1$ $ack$ events in $\gamma$.
Therefore, no more than $3$ of $v$'s groups can overlap $\gamma$ (as each group requires $c$ $ack$ events),
and therefore there are at most $3$ groups that both overlap $\gamma$ and deliver a message from $v$ to $u$ in this group.

We now lower bound the probability that $v$ sets $active$ to $false$ (and therefore only sends $nop$ messages to $u$) at the beginning
of all of these groups.
We consider two cases based on the value of $n_{\gamma}$. 
If $n_{\gamma} = 1$, then the fact that this group is calibrated only tells us that $n_v \geq 1$---which is not useful.
In this case, however, we note that the definition of the algorithm guarantees that $n_v \geq 2$, as it initializes $n_v$ to $2$ and these estimates never decrease. 
We can therefore crudely lower bound the probability that $v$ sets $active$ to $false$ in all overlapping groups,
by noting that it must be at least $(1-1/n_v)^3 \geq (1-1/2)^3 = 1/8 > 1/(64n)$---satisfying the lemma.

We now consider the case where $n_{\gamma} > 1$.
In this case, we leverage the definition of calibrated, 
which tells us that 
at the beginning of the first of these overlapping groups, $v$ has a network estimate $n_v \geq n_{\gamma}$, and that this remains
true for all overlapping groups as these estimates never decrease.
Therefore, the probability that $v$ delivers only $nop$ messages to $u$ during the first $c-1$ phases of $\gamma$ is at least:
$(1-1/n_v)^3 \geq (1-1/n_{\gamma})^3$.

Combining the above probability with the straightforward observation that $u$ is $active$ during $\gamma$ with probability at least $1/n$ (as $n$
is the largest possible network size estimate), yields the following probability that $\gamma$ is clean:

\begin{eqnarray*}
	(1/n) \cdot \prod_{v\in P_{\gamma}} (1-(1/n_v))^{3}  & \geq & (1/n) \cdot \prod_{v\in P_{\gamma}} (1-(1/n_{\gamma}))^{3} \\
	 & = & (1/n) \cdot \left( (1-(1/n_{\gamma}))^{3} \right)^{|P_{\gamma}|}\\
	 & = & (1/n) \cdot (1-(1/n_{\gamma}))^{3n_{\gamma}} \\
	 & \geq & (1/n) \cdot (1/4)^{(3n_{\gamma})/n_{\gamma}} \\
	 & \geq & 1/(64n),
	\end{eqnarray*}

\noindent as required by the lemma statement.
\end{proof}

We have established that {\em if} a group is calibrated then it has a good chance ($\approx 1/n$) of being clean and therefore ensuring termination.
To leverage this result, however, we must overcome two issues.
The first is proving that calibrated groups are sufficiently common in a given schedule.
The second is dealing with dependencies between different groups.
Assume, for example, we want to calculate the probability that at least one group from among a collection of target groups is clean.
Assume some node $u$ has a group that overlaps multiple groups in this collection.
If $u$ sets $active$ to $true$ in this group this reduces the probability of cleanliness for several groups in this collection.
In other words, cleanness probability is not necessarily independent between different target groups.

{\em On Good Target Groups.}
We overcomes these challenges by proving that any sufficiently long message schedule must contain a sufficient number
of calibrated and pairwise independent target groups.

To do so, let $\alpha$ be some message schedule generated by the scheduler that contains $qnx$ $ack$ events,
where $x=2c+1$ and $q= n+gn^2c\ln{n}$, for any constant $g\geq 512$.
Partition this schedule in $q$ {\em segments} each containing $nx$ $ack$ events.
Label these segments $s_1,s_2, ..., s_q$.

By Lemma~\ref{live:lem:x},
each segment $s_i$ contains an $x$-run for some node $u_i$.
Applying Lemma~\ref{live:lem:target},
it follows that this $x$-run contains at least one complete group for $u_i$.
We call this complete group the {\em target group}  for $s_i$,
and label it $t_i$. (If there are more than one complete groups for $u_i$ in the $x$-run,
then we set $t_i$ to the first such group in the run.)
Let $T=\{t_1,t_2,...,t_q\}$ be the complete set of these target groups.

We turn our attention to this set $T$ of target groups.
To study their useful for inducing termination, we will use the notion of {\em calibrated} introduced earlier, as well as the following
formal notion of {\em non-overlapping}:

\begin{definition}
Fix two target groups $t_i$ and $t_j$. 
We say $t_i$ and $t_j$ are {\em non-overlapping} if there does not exist a group that has at least one $recv$ event in $t_i$ and $t_j$.
If $t_i$ and $t_j$ are {\em not} non-overlapping, then we say they {\em overlap}.
\label{live:def:non}
\end{definition}

Our goal is to identify a subset of these target groups that are {\em good}---a property
which we define with respect to calibration and non-overlapping properties as follows:

\begin{definition}
Let $T' \subseteq T$ be a subset of the $q$ target groups.
We say the groups in $T'$ are {\em good} if: (1) every $t_i\in T'$ is calibrated; and (2)
for every $t_i, t_j\in T'$, where $i\neq j$, $t_i$ and $t_j$ are non-overlapping.
\label{live:def:good}
\end{definition}

Notice that both the calibration and non-overlapping status of groups are a function entirely of the message schedule.
Therefore, given a message schedule, 
we can partition it into segments and target groups as described above,
and label the status of these target groups without needing to consider the nodes' random bits.

To do so, we first prove a useful bound on the prevalence of calibration in $T$.
The core idea in the following is that every time a target group fails calibration, {\em all} nodes
increase their network estimates. 
Clearly, this can only occur $n$ times before all estimates are the maximum possible value of $n$,
after which calibration is trivial.
We then apply this result in making a more involved argument that on the frequency of good groups.

\begin{lemma}
At most $n$ groups in $T$ are {\em not} calibrated.
\label{live:lem:cal}
\end{lemma}
\begin{proof}
Fix some $t_i \in T$ that is not calibrated.
Let $P_i$ be the set of nodes that deliver at least one message to $u_i$ in the first $c-1$ phases of $t_i$.
By the definition of calibration, 
if $t_i$ is not calibrated, then at least one node $v\in P_i$ starts its relevant group with a network estimate $n_v < |P_i|$.

During the first $c-1$ phases of $t_i$, node $u_i$ receives a message from every node in $P_i$ (this is the definition of $P_i$). 
This means that by the start of the final phase of $t_i$, 
$u_i$'s network estimate is of size {\em at least} $|P_i|$. 
The message that $u_i$ sends in the final phase therefore will be labelled with this network size.
By the end of this final phase, all non-crashed processes will have received this estimate.
Therefore, all these processes will update their network size to be at least $|P_i|$ at the beginning of their next phases.

At this point, $|P_i|$ is now a minimum network size for the {\em entire} network. 
Therefore, if a subsequent group $t_j$ is not calibrated, then it must be the case that $|P_j| > |P_i|$,
and by the end of this group, the minimum network size for the entire network will increase to at least $|P_j|$.
Clearly, this increase process can happen at most $n$ times before the entire network has the maximum possible
network size of $n$, and every subsequent target group is trivially calibrated.
\end{proof}

\begin{lemma}
There exists a subset $T'\subseteq T$ such that the groups in $T'$ are good and $|T'| \geq gn\ln{n}$.
\label{live:lem:good}
\end{lemma}
\begin{proof}
Fix some $T$ as specified by the lemma statement.
We approach this proof from an algorithmic perspective.
That is, we describe below an algorithm that identifies a {\em good} subset $T'$ of $T$,
and then argue the subset produced by the algorithm is sufficiently large.

\begin{algorithm}
\begin{algorithmic}
\For{$i\gets 1$ to $q$}
	\If{$t_i$ is calibrated} 
		\State $\ell_i \gets good$
	\Else
		\State $\ell_i \gets bad$
	\EndIf
\EndFor
\For{$i\gets 1$ to $q$}
	\If{$\ell_i = good$}
		\For{$j \gets i+1$ to $q$}
			\If{$t_i$ overlaps $t_j$} $t_j \gets bad$ \EndIf
		\EndFor
	\EndIf
\EndFor
\State $T' \gets \{ t_i \mid \ell_i = good\}$
\end{algorithmic}
\end{algorithm}

We argue that $T'$ is good. 
First we note that by the definition of the algorithm, 
when a label gets set to $bad$ it can never again be changed back to $good$.

Next we note that if 
$\ell_i = good$ when $T'$ is defined in the final step,
then it could not be the case that $\ell_i$ was set to $bad$ in the first for loop,
as, by our first note, this would ensure that $\ell_i$ remained $bad$.
Therefore, $\ell_i$ must have been set to $good$ in the first for loop, indicating it is calibrated.

Now we consider overlaps.
If $\ell_i$ ends up $good$ then it must have been $good$ when the second for loop arrived at this value.
It follows that no preceding group overlaps $t_i$. 
During this iteration, the nested for loop will permanently set to $bad$ and succeeding target groups that $t_i$
overlaps.
Combined, it follows that if $\ell_i = good$ at this point, then for every $t_j$ that overlaps $t_i$ ($i \neq j$),
$\ell_j = bad$ before the second for loop completes.

We conclude that $T'$ is a good subset of $T$.
We now consider its sizes.
By Lemma~\ref{live:lem:cal},
we know that the first for loop marks at most $n$ groups $bad$ with the rest initialized to $good$.

Now consider an iteration $i$ of the second for loop that finds $\ell_i = good$.
We can bound the number of groups the inner for loop then sets to $bad$.
For each $v\neq u_i$, $v$ can have at most one group that delivers messages
to both $t_i$ and future groups. Call this group $g_v$.
Because $g_v$ delivers $c$ total messages,
the maximum number of future groups it can deliver messages to is at most $c-1$.
In the worst case, 
each $v\neq u_i$ therefore causes no more than $c - 1$ future groups to be labelled $bad$. %
There are $n-1$ total possible nodes,
so at most $(n-1)(c-1)$ future groups get labelled $bad$ for each $good$ group identified by the second for loop.
Therefore, if we divide these groups by $(n-1)(c-1) + 1$, we get a lower bound on the number of $good$ groups that remain:

\[  \frac{q-n}{(n-1)(c-1)+1} \geq \frac{q-n}{nc} =  \frac{(n+gn^2c\ln{n})-n}{nc} = gn\ln{n},\] %

\noindent as claimed by the lemma statement.
\end{proof}

The target groups in the set $T'$ identified by Lemma~\ref{live:lem:good} are calibrated and pairwise non-overlapping.
By Lemma~\ref{live:lem:prob}, each such group has a reasonable probability of being clean.
We will conclude our analysis by arguing that with high probability {\em at least one} will end up clean.

\begin{lemma}
Let $T' \subseteq T$ be a subset of the target groups $T$ such that the groups in $T'$ are good and $|T'| \geq gn\ln{n}$,
for some constant $g \geq 512$.
Then with high probability in $n$: at least one group in $T'$ is clean.
\label{live:lem:prob2}
\end{lemma}
\begin{proof}
Fix some $T'$ as specified by the lemma statement.
We describe the cleanliness of each $t_i\in T'$ with a random indicator variable $X_i$,
where $X_i = 1$ indicates $t_i$ is clean, and $X_i=0$ indicates it is not clean.
By Lemma~\ref{live:lem:prob}, we know that for each $t_i\in T'$: $\Pr(X_i = 1) \geq 1/(64n)$.  

We next argue that these random variables are independent.
To see why, notice that the {\em only} random choices made by a given node when reseting $active$ at the start of each group.
Each such choice is made with independent randomness: i.e., the choice for one group is independent from the choice made for any other group.
For any $t_i, t_j \in T'$, where $i\neq j$,
by the definition $T'$, there
are no groups that overlap both $t_i$ and $t_j$. 
Therefore, the random choices relevant to determine if $t_i$ is clean are distinct from the random choices that will determine if $t_j$ is clean.
It follows that $X_i$ and $X_j$ are independent. 

Let $Y=\sum_{t_i\in T'}X_i$ be the total number of clean groups.
It follows by linearity of expectation, Lemma~\ref{live:lem:prob}, and our assumption on the size of $T'$:

\[  E(Y) = E\left(  \sum_{t_i\in T'}X_i \right)  = \sum_{t_i\in T'} E(X_i) \geq |T'|/(64n)  \geq  (g/64)\ln{n}.\]

Because the $X$ indicators are independent, we can apply a Chernoff bound to concentrate around this expectation.\footnote{We use the following loose form of the bound that holds for $\mu = E(Y)$ when $0\leq \delta \leq 1$: 
$\Pr(Y \leq (1-\delta)\mu) \leq e^{-\frac{\delta^2\mu}{2}   }$.}
In particular, let $\mu = E(Y) \geq (g/64)\ln{n}$.
We bound the probability that $Y$ is a constant factor smaller than expected:

\begin{eqnarray*}
\Pr(Y \leq \mu/2)  & \leq & e^{-\frac{(1/2)^2 (g/64)\ln{n}}{2}   } \\
& = & e^{-(g/512)\ln{n}} \\
& = & 1/n^{g/512}\\
& \leq & 1/n
\end{eqnarray*}

Given our assumption on $g$, we know $\mu/2 \geq 1$.
Therefore, $Pr(Y \leq \mu/2)$ is less than or equal to the probability that no group is clean.
\end{proof}

We can now pull together the pieces to prove our main liveness result (Theorem~\ref{live:thm:main}):

\begin{proof}[Proof (of Theorem~\ref{live:thm:main}).]
We handled the case where $n=1$ at the beginning of the liveness analysis.
Here we consider only $n>1$, the case for which the above lemma hold.
To prove the first part of the theorem, 
fix some constant $g\geq 512$,
and define $q$ and $x$ as in the above definitions of segments and target groups.
Consider the first $qnx =  (n + gn^2c \ln {n}) \cdot n \cdot (2c+1) = \Theta(n^3\ln{n})$ $ack$ events of the message
schedule generated by the scheduler.
We can extract a set $T$ containing $q$ target groups from this prefix of the message schedule as described above.

By Lemma~\ref{live:lem:good}, there exists a subset $T' \subseteq T$ such that
the groups in $T'$ are good and $|T'| \geq gn\ln{n}$.
By Lemma~\ref{live:lem:prob2},
with high probability, at least one of these target groups is clean.
Finally, by Lemma~\ref{live:lem:decide}: if any group is clean, then by the end of that group every process
has either crashed, decided, or received a $decide$ message.

The second part of the theorem,
which addresses termination in the limit,
we first note that if we continually apply the argument from Lemma~\ref{live:lem:prob2}
to fresh batches of groups, the probability that we do not generate a clean group approaches $0$ in the limit.
Combined with Lemma~\ref{live:lem:decide}, this provides the needed probabilistic termination condition.
\end{proof}

}

\onlyLong{\vspace{-0.3cm}}
\subsection{Removing the Assumption of Unique IDs}  \label{sec:ids}
\onlyLong{\vspace{-0.3cm}}

The consensus algorithm described in this section assumes unique IDs. 
We now show how to eliminate this assumption by describing a strategy 
that generates unique IDs w.h.p., and discuss how to use this as a subroutine in our consensus algorithm.

We make use of a simple tiebreaking mechanism as follows: Each node $u$ proceeds by iteratively extending a (local) random bit string that eventually becomes unique among the nodes. 
Initially, $u$ broadcasts bit $b_1$, which is initialized to $1$ (at all nodes), and each time $u$ samples a new bit $b$, it appends $b$ to its current string and broadcasts the result. 
For instance, suppose that $u$'s most recently broadcast bit string is $b_1\dots b_i$. Upon receiving $ack(b_1\dots b_i)$, node $u$ checks if it has received a message identical to $b_1\dots b_i$. 
If it did not receive such a message, then $u$ adopts $b_1\dots b_i$ as its ID and stops. 
Otherwise, some distinct node must have sampled the same sequence of bits as $u$ and, in this case, the ID $b_1\dots b_i$ is considered to be already taken. 
(Note that nodes do not take receive events for their own broadcasts.) 
Node $u$ continues by sampling its $(i+1)$-th bit $b_{i+1}$ uniformly at random, and then broadcasts the string $b_1\dots b_i b_{i+1}$, and so forth. 
\onlyLong{
\begin{algorithm}
\begin{algorithmic}[1]
\State \underline{Initialization:}
\State $b_1 \gets 1$; $R \gets \emptyset$; $i=1$
\State $\textbf{bcast}(b_1)$
\item[]
\State \underline{On Receiving $ack(b_1\dots b_i)$}
\If{ $(b_1 \dots b_i)  \notin R$}
  \State $id_u \gets (b_1 \dots b_i)$ \label{line:fixid}
  \State adopt $id_u$ as ID and terminate \label{line:breakloop}
\EndIf
\State $i\gets i+1$
\State sample bit $b_i$ uniformly at random
\State $\textbf{bcast}(b_1 \dots b_i)$
\item[]
\State \underline{On Receiving message $(b_1'\cdots b_j')$, ($j\ge 1$):} 
\State \textbf{if} $u$ has not yet assigned $id_u$ \textbf{then} add $(b_1'\cdots b_j')$ to $R$
\end{algorithmic}
\caption{Generating unique IDs using randomized tiebreaking. Code for node $u$.}
\label{alg:ids}
\end{algorithm}

We first show that the algorithm is safe in the sense that no two nodes ever assign themselves the same ID:

\begin{lemma} \label{lem:IDsDistinct}
  Suppose that nodes $u$ and $v$ both terminate Algorithm~\ref{alg:ids}.
  Then it holds that $id_u \ne id_v$.
\end{lemma}
\begin{proof}
Consider an execution $\alpha$ and the corresponding message schedule $msg[\alpha]$.
Suppose, in contrary, that $id_u = id_v$. 
Let $r_u$ and $r_v$ denote the number of acks that $u$ respectively $v$ receive before assigning an ID and, without loss of generality, assume $r_u \le r_v$.
Clearly, if $r_u < r_v$, then $id_v$ is at least one bit longer than $id_u$, thus $id_v > id_u$. 
Now suppose that $r_u = r_v$, i.e., both $u$ and $v$ receive the same number of $acks$.
Let $t_u$ and $t_v$ be the events in $msg[\alpha]$ where $u$ and $v$ receive their respective $r_u$-th $ack$ and, without loss of generality, assume that $t_u$ precedes $t_v$ in $msg[\alpha]$.
By assumption, $u$ was non-faulty until receiving its $ack$ in event $t_u$ and hence $v$ must have received $u$'s broadcast message $(id_u)$ before receiving its own $ack$ in step $t_v$.
Since $u$ and $v$ have generated the same bits by assumption, the if-conditional ensures that $v$ samples at least one additional bit compared to $u$, providing a contradiction.
\end{proof}
Next, we show liveness in Lemma~\ref{lem:IDsTermTime} by arguing that each node receives an ID within its first $O(\log n)$ broadcast events with high probability. 
\begin{lemma} \label{lem:IDsTermTime}
  With high probability, each node broadcasts at most $O(\log n)$ times before choosing an ID in Line~\ref{line:fixid} of Algorithm~\ref{alg:ids}.
\end{lemma}
\begin{proof}
Consider an execution $\alpha$ and assume, towards a contradiction, that a node $u$ executes at least $\lceil 4\log_2 n\rceil+2$ broadcasts. 
Let $(b_1 \dots b_{\lceil 4\log_2 n\rceil+1})$ be the $(\lceil 4\log_2 n\rceil+1)$-length bit string broadcast by $u$, 
and let $t$ be the event where $u$ receives $ack(b_1\dots b_{\lceil 4\log_2 n\rceil+1})$ for the corresponding broadcast. 
By assumption, $u$ does not pass the if-condition in event $t$ and thus there is a set of nodes $W$ ($u \notin W$) that also broadcast bit strings of length $\lceil 4\log_2 n\rceil + 1$ and whose messages are received by $u$ before event $t$.
While the first bit $b_1$ is initialized to $1$ by every node, the string $b_2\dots b_{\lceil 4\log_2 n\rceil+1})$ corresponds to a uniform random sample from a range of size at least $n^4$. 
The probability that $v$ has sampled precisely the same $\lceil 4\log_2 n\rceil$ bits as $u$ (and hence broadcast $(b_1 b_2\dots b_{\lceil 4\log_2 n\rceil+1})$) is at most $\frac{1}{n^4}$.
Taking a union bound over all other nodes in $W$ and over all possible choices for $u$, 
shows that all nodes will execute at most $\lceil 4\log_2 n+1\rceil$ broadcasts with high probability.
\end{proof}
From the previous two lemmas, we obtain the following result:
}
\onlyShort{
In the full paper \cite{fullpaper}, we prove the following result and describe how to combine it with our consensus algorithm:
}
\begin{theorem} \label{thm:ids}
  Consider an execution $\alpha$ of the tiebreaking algorithm.
  Let $t_u$ be an event in the message schedule $msg[\alpha]$ such that node $u$ is scheduled for $\Omega(\log n)$ ack events before $t_u$.
  Then, for each correct node $u$, it holds that $u$ has a unique ID of $O(\log n)$ bits with high probability at $t_u$.
\end{theorem}
\onlyLong{
Equipped with Theorem~\ref{thm:ids} we can execute the consensus algorithm in networks without unique IDs, by instructing each node $u$ to first execute Algorithm~\ref{alg:ids}, while locally buffering all messages received from nodes already executing the consensus algorithm; however, $u$ does not yet process these messages. Once $u$ obtains an ID, it performs the initialization step of the consensus algorithm and locally simulates taking receive steps for all previously buffered messages.
}

\newcommand{\act}{\textsf{Active}}
\newcommand{\unsupp}{\textsf{Unsupp}}

\section{Almost-Everywhere Agreement}
\onlyShort{\vspace{-0.2cm}}
\label{fast}

In the previous section,
we showed how to solve consensus in $O(n^3\log{n})$ events.
Here we show how to improve this bound by a near linear factor by loosening the agreement guarantees.
In more detail, we consider a weaker variant of consensus,  introduced in \cite{dwork:1988}, called \emph{almost-everywhere agreement}.
This variation relaxes the agreement property of consensus such that
$o(n)$ nodes are allowed to decide on conflicting values, as long as the remaining nodes all decide the same value.
For many problems that use consensus as a subroutine, this relaxed agreement property is sufficient.

In more detail, we present an algorithm for solving almost-everywhere agreement in the abstract MAC layer model when nodes start with arbitrary (not necessarily binary) input values.
The algorithm consists of two phases; see Algorithm~\ref{alg:aea} for the pseudo code. %

\noindent\textbf{Phase 1:}
In this phase, nodes try to obtain an estimate of the network size by performing local coin flipping experiments.
Each node $u$ records in a variable $X$ the number of times that its coin comes up tails before observing the first heads.
Then, $u$ broadcasts its value of $X$ once, and each node updates $X$ to the highest outcome that it has seen until it receives the $ack$ for its broadcast.
We show that, for all nodes in a large set called $EST$, variable $X$ is an approximation of $\log_2(n)$ with an additive $O(\log \log n)$ term by the end of Phase~1, and hence $N := 2^{X}$ is a good approximation of the network size $n$ for any node in $EST$.
 
\noindent\textbf{Phase 2:}
Next, we use $X$ and $N$ as parameters of a randomly rotating leader election procedure.
Each node decides after $T = \Theta(N \log^3 (N) \log\log(N))$ \emph{rounds}\onlyShort{. }\onlyLong{, where each round corresponds to one iteration of the for-loop in Algorithm~\ref{alg:aea}.}
(Note that due to the asynchronous nature of the abstract MAC layer model, different nodes might be executing in different rounds at the same point in time.)
We now describe the sequence of steps comprising a round in more detail: A node $u$ becomes active with probability $1/N_u$ at the start of each round.\footnote{We use the convention $N_u$ when referring to the local variable $N$ of a specific node $u$.} 
If it is active, then $u$ samples a random rank $\rho$ from a range polynomial in $X_u$, and broadcasts a message $\langle r, \rho, val \rangle$ where $val$ refers to its current consensus input value.
To ensure that the scheduler cannot derive any information about whether a node is active in a round, inactive nodes simply broadcast a dummy message with infinite rank.
While an (active or inactive) node $v$ waits for its $ack$ for round $r$, it keeps track of all received messages and defers processing of a message sent by a node in some round $r'>r$ until the event in which $v$ itself starts round $r'$.
On the other hand, if a received message was sent in $r'<r$, then $v$ simply discards that late message as it has already completed $r'$.
Node $v$ uses the information of messages originating from the same round $r$ to update its consensus input value, if it receives such a message from an active node that has chosen a smaller rank than its own. 
(Recall that inactive nodes have infinite rank.)
After $v$ has finished processing the received messages, it moves on the next round. 

We first provide some intuition why it is insufficient to focus on a round $r$ where the ``earliest'' node is also active: 
Ideally, we want the node $w_1$ that is the first to receive its $ack$ for round $r$ to be active \emph{and} to have the smallest rank among all active nodes in round $r$, as this will force all other (not-yet decided) nodes to adopt $w_1$'s value when receiving their own round $r$ $ack$, ensuring a.e.\ agreement.
However, it is possible that $w_1$ and also the node $w_2$ that receives its round $r$ $ack$ right after $w_1$, are among the few nodes that ended up with a small (possibly constant) value of $X$ after Phase~1. %
We cannot use the size of $EST$ to reason about this probability, as some nodes are much likelier to be in $EST$ than others, depending on the schedule of events in Phase~1.
In that case, it could happen that both $w_1$ and $w_2$ become active and choose a rank of $1$.
Note that it is possible that the receive steps of their broadcasts are scheduled such that roughly half of the nodes receive $w_1$'s message before $w_2$'s message, while the other half receive $w_2$'s message first. 
If $w_1$ and $w_2$ have distinct consensus input values, then it can happen that both consensus values gain large support in the network as a result. 

To avoid this pitfall, we focus on a set of rounds where {all} nodes \emph{not} in $EST$ have already terminated Phase~2 (and possibly decided on a wrong value): from that point onwards, only nodes with sufficiently large values of $X$ and $N$ keep trying to become active.
We can show that every node in $EST$ has a probability of at least $\Omega(1/(n\log n))$ to become active and a probability of $\Omega(1/\log n)$ to have chosen the smallest rank among all nodes that are active in the same round.
Thus, when considering a sufficiently large set of (asynchronous) rounds, we can show that the event, where the first node in $EST$ that receives its $ack$ in round $r$ becomes active and also chooses a rank smaller than the rank of any other node active in the same round, happens with probability $1 - o(1)$.

\begin{algorithm}[t]
\begin{algorithmic}[1]
\State $val \gets $ consensus input value

\State \Comment{Phase 1}
\State initialize $X \gets 0$; $R \gets \emptyset$
\While{$flip\_coin()=heads$}
  \State $X \gets X + 1$
\EndWhile
\State $\textbf{bcast}(X)$ \label{line:fstbcast}
\While{waiting for $ack$}
  \State add received messages to $R$
\EndWhile
\State $X \gets \max(R \cup \{X\})$
\State $N \gets 2^X$

\State \Comment{Phase 2}
\State $T \gets \lceil c N \log^3(N)\log\log(N) \rceil$, where $c$ is a sufficiently large constant. \label{line:t}
\State initialize array of sets $R[1],\dots,R[T] \gets \emptyset$
\For{$i \gets 1,\dots, T$} \Comment{Start of round $i$ at $u$}
  \State $u$ becomes active with probability $\tfrac{1}{N}$
  \If{$u$ is active}
    \State $\rho \gets$ unif.\ at random sampled integer from $[1,X^4]$
  \Else
    \State $\rho \gets \infty$
  \EndIf
  \State $\textbf{bcast}(\langle i, \rho, val\rangle)$
  \While{waiting for $ack$}
    \State add received messages to $R[i]$
  \EndWhile
  \For{\textbf{each} message $m=\langle i', \rho', val'\rangle \in R[i]$}
    \If{$i'= i$ and $\rho' < \rho$} \Comment{Received message from node with smaller rank}
      \State $val \gets val'$
    \ElsIf{$i'> i$} \Comment{Received message from node active in future round}
      \State add $m$ to $R[i']$
    \Else
      \State discard message $m$
    \EndIf
  \EndFor
\EndFor
\State decide on $val$
\item[]
\end{algorithmic}
\caption{Almost-everywhere agreement in the abstract MAC layer model. Code for node $u$.}
\label{alg:aea}
\end{algorithm}
\onlyLong{In the remainder of this section, we will formalize the above discussion by proving the following main theorem regarding this algorithm:}
\onlyShort{In the full paper \cite{fullpaper}, we formalize the above discussion by proving the following main theorem regarding this algorithm:}

\begin{theorem} \label{thm:aea}
With high probability, the following two properties are true of our almost-everywhere consensus algorithm: (1) within $O(n^2 \log^4 n\cdot\log\log n)$ scheduled $ack$ events, every node has either crashed, decided, or will decided after it is next scheduled;
(b) all but at most $o(n)$ nodes that decide, decide the same value.%
\end{theorem}

\onlyLong{
We begin our proof of Theorem~\ref{thm:aea} by analyzing the properties of variables $N$ and $X$.
We say that a \emph{node $u$ fails in round $r$} if $u$ performs its round $r$ broadcast in some event, but crashes before receiving its corresponding $ack$; otherwise, we say \emph{$u$ is alive in $r$}. Note that there is no guarantee about which nodes receive a failing node's final broadcast.

\begin{lemma} \label{lem:est}
There exists a set of nodes $EST$ of size at least $\brackets{1 - O\brackets{\frac{\log\log n}{\log n}}}n - f$, such that the following hold with probability at least $1 - o(1)$:
\begin{compactitem}
\item[(a)] for all $u \in EST$, when $u$ receives its $ack$ for its first broadcast, it holds that
\begin{align}
  \frac{n}{\log_2 n} \le N_u \le n\log n\ \ \text{and}\ \log_2 n - \log_2(\log n) \le X_u \le \log_2 n + \log_2 \log n; \label{eq:nstar}
\end{align}
\item[(b)] for all $v \notin EST$, we have $N_v \le \frac{n}{2\log_2 n}$.
\end{compactitem}
\end{lemma}

Our proof of Lemma~\ref{lem:est} requires a technical result on the distribution of observed coin flips.
\begin{claim} \label{cl:nstar}
Consider any set $S$ of at least $\frac{2n\log\log n}{\log n}$ correct nodes and let $X^* = \max\{ X_u \mid u \in S \}$, where $X_u$ refers to $u$'s variable before node $u$ receives any messages in Phase~1.
It holds with probability at least $1 - O(1/\log n)$ that
$\log_2 n - \log_2(\log n) \le X^* \le \log_2 n + \log_2 \log n.$
\end{claim}
\noindent\textit{Proof of Claim~\ref{cl:nstar}.}
Observe that $X_u$ is geometrically distributed with parameter $\frac{1}{2}$ and hence
\[
  \Prob{ X_u \ge \log_2 n + \log_2\log n } \le 2^{-\log_2 n - \log_2\log n} = \frac{1}{n\log n}.
\]
Taking a union abound over all nodes in $S$ and noting that $|S|\le n$, implies that $X^* \le  \log_2 n + \log_2 \log n$ with probability at least $1 - 1/\log n$, proving the upper bound.

For the lower bound, we first bound the probability that the estimate of a single node $u \in S$ is above the required threshold. We get
\[
  \Prob{ X_u \!\ge\! \log_2 n - \log_2 (\log n)}
    \ge \Prob{ X_u \!=\! \log_2 n - \log_2 (\log n)}
    = 2^{-\log_2 n + \log_2 (\log n) - 1} = \frac{ \log n}{2n},
\]
where the second last equality follows from the properties of the geometric distribution.
Considering the complementary event, namely that $X_u$ is below the threshold, and taking a union bound over the set $S$, yields
\[
  \Prob{ \forall u \in S\colon X_u < \log_2 n - \log_2 (\log n)}
    \le \brackets{1 -  \frac{ \log n}{2n}}^{2n\log\log n/\log n}
    \le \exp\brackets{ - \frac{ 2n \log n\log \log n}{2n \log n}},
\]
thus completing the proof of Claim~\ref{cl:nstar}. \qed

\begin{proof}[Proof of Lemma~\ref{lem:est}]
We note that the values $N$ are powers of $2$ and, since any node not in $EST$ must have a value of $X$ strictly smaller than for any node that is in $EST$, Part~(b) follows.
Thus we focus on (a) in the remainder of the proof.

To obtain a lower bound on the size of $EST$, we define the set $S$ in the premise of Claim~\ref{cl:nstar} to consist of the first $\left\lceil \frac{2n\log\log n}{\log n} \right\rceil$ nodes that receive the $ack$ for their broadcast in Phase~1 of the algorithm.
Let $\bar{S}$ be the set of alive nodes that are not in $S$.
Then, all nodes in $\bar{S}$ are guaranteed to receive the maximum value broadcast by nodes in $S$ before completing Phase~1.
Observe that any node $u \in \bar{S}$ must have $N_u \le n\log n$  by instantiating Claim~\ref{cl:nstar} with set $\bar{S}$.
Since $EST$ contains at least all nodes in $\bar{S}$, the lemma follows.
\end{proof}

We now focus on Phase~2 of the algorithm which is conceptually structured into asynchronous rounds, where each round consists of one iteration of the for-loop of Algorithm~\ref{alg:aea}. When talking about some event $E$ in round $r$ that concerns a set of nodes $U$, we refer to the collection of events in the message schedule where the nodes in $U$ execute the corresponding events. 
We say that $u \in EST$ is the \emph{earliest node in round $r$}, if $u$ receives its $ack$ for its round $r$ broadcast before all other nodes in the message schedule.
Note that which node is the earliest depends on the scheduler and can change from round to round.

\begin{lemma} \label{lem:active}
With probability $1 - O(1/\log n)$, there exists a set $\Gamma$ of at least $\Omega\brackets{n\log^2 n \log \log n}$ rounds in which no node crashes and where the following hold:
\begin{compactitem}
\item[(a)] in every round $r \in \Gamma$, at most $4\log_2 n$ nodes in $EST$ become active in $r$, and
\item[(b)] all nodes in $EST$ remain undecided until the last round of $\Gamma$.
\end{compactitem}
\end{lemma}
\begin{proof}
  For Part~(a), recall from Lemma~\ref{lem:est}.(a) that each node $u \in EST$ becomes active with probability $1/N_u \le \frac{\log_2 n}{n}$ and hence the expected number of active nodes is at most $\log_2 n$.
Since nodes become active independently, an application of a standard Chernoff bound \cite{mitzenmacher:2004} shows that at most $4\log_2 n$ nodes in $EST$ become active with high probability.

We now consider Part~(b).
We know that the number of rounds executed by any node $v \notin EST$ is at most
\begin{align}
  T_v = \left\lceil c N_v \log^3(N_v) \log\log(N_v) \right\rceil
  &\le \frac{c n}{2\log_2 n}\log^3\brackets{\frac{n}{2\log_2 n}} \log\log\brackets{\frac{n}{2\log_2 n}} +1 \tag{by Lem.~\ref{lem:est}.(b)} \\
  &\le \frac{19}{36}\frac{c n}{\log n}\log^3\brackets{\frac{n}{2\log n}} \log\log\brackets{\frac{n}{2\log n}} \notag\\
  &\le \frac{19}{36} c n\log^2 n \log \log n. \label{eq:tv}
\end{align}
On the other hand, Lemma~\ref{lem:est}.(a) tells us that any $u \in EST$ executes at least
\[
  T_u \ge c N_u \log^3(N_u)\log\log(N_u)
     \ge \frac{c n}{\log_2 n}\log^3\brackets{\frac{n}{\log_2 n}}\log\log\brackets{\frac{n}{\log_2 n}}
\]
rounds.
For sufficiently large $n$, it holds that $\log^3\brackets{\frac{n}{\log_2 n}} \ge \frac{5}{6}\log^3 n$ and similarly $\log\log\brackets{\frac{n}{\log_2 n}} \ge \frac{5}{6}\log\log n$.
Thus, simplifying the right-hand side in the above inequality yields
\[
  T_u \ge \tfrac{25}{36}c n \log^2 n\log\log n.
\]
Recalling \eqref{eq:tv}, it follows that there is a set $\Gamma_f$ of at least
$T_u - T_v \ge \frac{c}{6}n\log^2 n\log\log n$ rounds where only nodes in $EST$ execute the code in the for-loop.
Since nodes can fail in at most $n-1$ rounds of the algorithm,
it follows that there exists a subset $\Gamma\subseteq\Gamma_f$ of size at least $\Omega\brackets{n\log^2 n\log\log n}$, as required.
\end{proof}

\begin{lemma} \label{lem:probMinRank}
Suppose that there is a set $EST$ as stated in Lemma~\ref{lem:est} and assume that the set of rounds $\Gamma$ implied by Lemma~\ref{lem:active} exists.
Then there exists a round $r \in \Gamma$ such that, with probability $1 - O(1/\log n)$, the earliest node is alive in $r$, becomes active, and has the minimum rank.
\end{lemma}

\begin{proof}
Below, we restrict our attention to the set of rounds $\Gamma$ where only nodes in $EST$ participate.
We will first lower bound the probability that an active node has the lowest rank among all nodes active in round $r \in \Gamma$.

Condition on the event that the earliest node $u$ is active in $r$.
Let $q$ be the probability that $u$ chooses a unique minimum rank among active nodes and consider the threshold $L=\log_2 n - \log_2(\log n)$.
Recall from \eqref{eq:nstar}, that all nodes in $EST$ choose their rank from a range $[1,\ell]$ where $\ell \ge L^4$.
Let ``$\rho_u$ min'' be the event that $u$ chooses the smallest rank in this round.
We get
\begin{equation}
  \label{eq:q}
  q = \Prob{ \text{$\rho_u$ min} \mid u \in \act} \ge \Prob{ \text{$\rho_u$ min} \mid  u \in \act, \rho_u \le L^4} \cdot \Prob{ \rho_u \le L^4 \mid u \in \act}.
\end{equation}
We now prove that $q \ge \frac{1}{20\log_2 n}$.
For all active $v \in EST$, it holds that $\rho_v \le (\log_2 n + \log_2 \log n)^4 \le 2\log_2^4 n$.
Together with the fact that $L \ge \tfrac{1}{2} \log_2 n$, this implies that
\begin{align}
  \label{eq:q2}
  \Prob{ \rho_u \le L^4  \mid u \in \act } \ge \frac{L^4}{2\log_2^4 n} \ge \frac{1}{4}.
\end{align}
Next, we will derive a bound on $\Prob{ \text{$\rho_u$ min} \mid  u \in \act, \rho_u \le L^4}$.
Lemma~\ref{lem:active}.(a) tells us that there are at most $4\log_2 n$ active nodes in any given round $r \in \Gamma$. Consider some active node $v$. If we condition on all nodes  choosing their rank from the range $[1,L^4]$, the probability that all nodes choose \emph{distinct} ranks from the rank of $v$ must be at least $\left(1 - \frac{1}{L^4}\right)^{4\log_2 n}$. 
In that case, a union bound over the active nodes implies that, for the event $\textsf{dist}$, which occurs when all nodes have unique ranks, we get
\begin{align}
  \Prob{\textsf{dist} \mid u \in \act, \forall v \in \act\colon \rho_v \le L^4} \ge \left(1 - \frac{1}{L^4}\right)^{4\log_2^2 n} \ge 1 - O\brackets{1/\log^2 n}. \label{eq:probdist}
\end{align}
Moreover, conditioning on the event that all active nodes choose ranks from $[1,L^4]$ does not increase the probability of $u$ choosing the smallest rank, which tells us that 
\begin{align*}
  \Prob{ \text{$\rho_u$ min} \mid u \in \act, \rho_u \le L^4}
    &\ge \Prob{ \text{$\rho_u$ min} \mid u \in \act, \forall v\in \act\colon \rho_v \le L^4} \\
    &\ge \Prob{ \text{$\rho_u$ min} \mid u \in \act, \forall v\in \act\colon \rho_v \le L^4,\textsf{dist}} \brackets{1 - O\brackets{\frac{1}{\log^2 n}}},\\
\intertext{
where the last inequality follows from \eqref{eq:probdist}.
Conditioned on event $\textsf{dist}$ and the premise of the lemma of having at most $4\log_2 n$ active nodes, the probability of $u$ picking the smallest rank is at least $1/4\log_2 n$, i.e.,
  $\Prob{ \text{$\rho_u$ min} \mid u \in \act, \forall v\in \act\colon \rho_v \le L^4,\textsf{dist}} 
  \ge \frac{1}{4\log_2 n}.$
It follows that
}
  \Prob{ \text{$\rho_u$ min} \mid u \in \act, \rho_u \le L^4}
    &\ge \frac{1}{4\log_2 n}  \brackets{1 - O\brackets{\frac{1}{\log^2 n}}} \\
    &\ge \frac{1}{5\log_2 n},
\end{align*}
Plugging the above bound and \eqref{eq:q2} into the right-hand side of \eqref{eq:q} shows that
$q \ge \frac{1}{20\log_2 n}$.
 
Conditioned on Lemma~\ref{lem:est}.(a), we know that every node in $EST$, and in particular, the earliest node $u$, has probability at least $\frac{1}{n\log n}$ of being active in any single round $r \in \Gamma$.
We have
\begin{align*}
  \Prob{ \text{$\rho_u$ min} \wedge u \in \act }
      = \Prob{ \text{$\rho_u$ min} \mid u \in \act}\cdot \Prob{  u \in \act } 
    & \ge \frac{q}{n\log n}
      \ge \frac{1}{20n\log_2^2 n},
\end{align*}
for any round $r \in \Gamma$ and the respective earliest node $u$ in $r$.

Recalling that $\Gamma$ comprises $\Omega(n\log^2 n\log \log n)$ rounds, it follows that the event that, for none of the rounds in $\Gamma$, the earliest node becomes the smallest ranked active node, happens with probability at most
\[
\brackets{1-\frac{1}{20n\log_2^2 n}}^{|\Gamma|} \le \exp\brackets{-\frac{|\Gamma|}{20n\log_2^2 n}} = O\brackets{\frac{1}{\log n}}.
\]
\end{proof}

\paragraph*{Proof of Theorem~\ref{thm:aea}:}
Validity follows since any value written to variable $val$ was the input value of some node.

For termination, notice the number of rounds executed by any node $u$ depends on the value of $T_u = O(N_u\log^3 (N_u) \log\log(N_u))$ in Phase~2.
From Claim~\ref{cl:nstar}, we know that $N_u \le n \log n$ for all nodes $u$ with probability $1 - o(1)$ and hence the maximum number of rounds executed by any node $u$ is $O(n \log^4 n\log\log n)$, which results in the same bound for the total number of broadcasts by $u$. Taking into account that there are $n$ nodes, the claimed termination bound follows. 

Conditioned on the properties of set $EST$ (cf.\ Lemma~\ref{lem:est}), we now show that almost all nodes decide on a common value. 
From  Lemma~\ref{lem:probMinRank} we know that with probability $1 - o(1)$, there is a set $\Gamma$ containing a round $r\in \Gamma$, in which the earliest node $u$ is active, non-faulty, and has the minimum rank.
Let $t'$ be the event when $u$ receives the corresponding $ack$ for its round $r$ broadcast message $m_u$ carrying $val_u$. 
By Lemma~\ref{lem:active}.(b), we know that every node $v \in EST$ is performing all rounds in $\Gamma$ and hence will receive $u$'s message $m_u$ in some receive event $t_v'$ that precedes $t'$ in the message schedule. 
Moreover, since $u$ was the earliest node in round $r$, it follows that event $t_v'$ must be part of some round $r' \le r$ (at $v$) and in particular must occur before $v$ receives its $ack$ for round $r$.
If $r'<r$, then $v$ defers the processing of message $m_u$ until $v$ reaches round $r$; otherwise, if $r'=r$, then $v$ adopts $u$'s value when it receives its $ack$.
By Lemma~\ref{lem:active}.(b), the nodes in $EST$ execute all rounds of $\Gamma$ and hence all of them will adopt $val_u$ when receiving their $ack$ in round $r$. %

To complete the proof, we will argue that no node in $EST$ changes its value after round $r$. 
For the sake of a contradiction, suppose that there is some $w \in EST$ that adopts some value $z \ne val_u$ during an $ack$ event $t_w'$ in some round $r_w > r$.
Moreover, assume that $t_w$ is the earliest such event in the message schedule that is causally influenced by $u$'s round $r$ broadcast event $t$. 
Since $u$ has the smallest rank in $r$, it follows that $w$ must have received a message $\langle r',\rho',x \rangle$, which was sent by some node $u'$ during its round $r' \ne r$.
First, observe that if $r' < r$, then also $r' < r_w$ and hence $w$ would have discarded that message in event $t_w$. 
Now consider the case $r' > r$. Since only nodes in $EST$ perform broadcasts during the rounds in $\Gamma$, it follows that $u' \in EST$, and hence by the above argument we know that $u'$ must have broadcast some $x \ne val_u$ after having adopted $val_u$ in its round $r$.
This means that $u'$ updated its value in some event after round $r$ but before $t_w$, contradicting the assumption that $t_w$ was the earliest event (after round $r$) in the message schedule where such an update occurred.
It follows that at least $|EST| - f = n\brackets{1 - O\brackets{\frac{\log\log n}{\log n}}} - f$ nodes decide on a common value.

When applying Lemmas~\ref{lem:est}, \ref{lem:active}, and \ref{lem:probMinRank} in the argument above, we condition on events each of which happens with probability $1 - o(1)$.
Hence we can remove the conditioning while retaining a probability of success of $1 - o(1)$. \qed

}

\onlyShort{\vspace{-0.4cm}}
\section{Lower Bound}
\onlyShort{\vspace{-0.2cm}}
\label{sec:lower}

We conclude our investigation by showing a separation between the abstract MAC layer model and the related asynchronous message passing model.
In more detail, we prove below that fault-tolerant consensus with constant success probability
is impossible in a variation of the asynchronous message passing model where nodes are provided only a constant-fraction approximation of the network size
and communicate using (blind) broadcast.
This bounds holds even if we assume no crashes and provide nodes unique ids from a small set.
Notice, in the abstract MAC layer model,
we solve consensus with broadcast under the harsher constraints of {no} network size information, no ids, and crash failures.
The difference is the fact that the broadcast primitive in the abstract MAC layer model includes an acknowledgment.
This acknowledgment is therefore revealed to be the crucial element of the our model that allows algorithms to overcome lack of network information.
We note that this bound is a generalization of the result from~\cite{abboud:2008},
which proved deterministic consensus was impossible under these constraints.
\onlyLong{In the proof of the theorem, we show that, for any given randomized algorithm we can construct scenarios that are indistinguishable for the nodes, thus causing conflicting decisions.  
}
\onlyShort{
In the full paper \cite{fullpaper}, we show that, for any given randomized algorithm we can construct scenarios that are indistinguishable for the nodes, thus causing conflicting decisions.  
}

\begin{theorem} \label{thm:asyncImposs}
Consider an asynchronous network of $n$ nodes that communicate by broadcast and suppose that nodes are unaware of the network size $n$, but have knowledge of an integer that is guaranteed to be a $2$-approximation of $n$. 
No randomized algorithm can solve binary consensus with a probability of success of at least $1 - \epsilon$, for any constant $\epsilon< 2 - \sqrt{3}$. This holds even if nodes have unique identifiers chosen from a range of size at least $2n$ and all nodes are correct. 
\end{theorem}
\onlyLong{
\begin{proof}
In our proof we construct admissible executions by restricting ourselves to schedules that are infinite sequences of layers (cf.\ \cite{moses:2002}).
For a given set of nodes $S$, we define a \emph{layer} $L(S)$ to consist of an arbitrarily ordered sequence of nodes in $S$, say $\langle u_1,\dots,u_k \rangle$, followed by a sequence of sets of received messages $\langle M_1,\dots,M_k\rangle$, where $M_i$ denotes the set of messages received by node $u_i$. 
Layer $L(S)$ defines a schedule where each $u_i$ takes a compute step (in the given order), in which it can perform some local computation and broadcast a message. 
We conclude the layer by scheduling each $u_j \in S$ to take sufficiently many receive steps to ensure that all messages in $M_j$ are delivered. We restrict the sets $M_j$ such that each message $m \in M_j$ must have been broadcast in $L(S)$ or some layer preceding $L(S)$ in the schedule. 

Assume, towards a contradiction, that there is a randomized consensus algorithm that succeeds with probability $\ge 1 - \epsilon$.
Consider the $n$-node clique network $H_0$ of nodes $u_1,\dots,u_n$ where each node is equipped with some arbitrary unique identifier and all nodes start with consensus input $0$. Moreover, nodes are given the network size estimate $2n$. By a slight abuse of notation, we use $H_0$ to refer to both, the network and the set of nodes in the network. We specify the schedule $\sigma_0$ to be the infinite sequence $\langle L(H_0),L(H_0),\dots\rangle$ where layer $L(H_0)$ is such that all broadcasts by nodes in $H_0$ are received by all nodes in $H_0$ in the very same layer in which they are sent. 
Since $\sigma_0$ results in an admissible execution according to the asynchronous broadcast model, there exists a fixed integer $t_0$ such that all nodes in $H_0$ have decided with probability at least $1 - 1/n$ within the first $t_0$ steps of $\sigma_0$.
Validity and agreement tell us that, if nodes decide in the $t_0$-step prefix $\sigma_0'$ of $\sigma_0$, their decision must be on $0$ with probability at least $1 - \epsilon$.

Similarly, we define a schedule $\sigma_1 = \langle L(H_1),L(H_1),\dots\rangle$ on a network $H_1$ of $n$ nodes where all nodes start with input $1$, a network size estimate of $2n$, and nodes are given a set of unique IDs disjoint from the IDs used for $H_0$. By a similar argument as above, there is an integer $t_1$ such that the algorithm ensures a common decision on $1$ with probability at least $1 - \epsilon$, conditioned on nodes deciding within $t_1$ steps (which itself is bound to happen with probability $\ge 1 - 1/n$); we denote the corresponding schedule prefix by $\sigma_1'$.

Now, we consider the clique network $G$ on the set of nodes $H_0 \cup H_1$ where nodes in $H_0$ have input $0$, nodes in $H_1$ start with input $1$, and the same set of IDs are assigned as above.
Here nodes are given the same network size estimate, i.e., $2n$, as in networks $H_0$ and $H_1$, which unbeknownst to them is the actual network size of $G$.
We define an infinite ``synchronous'' schedule $\sigma_2$ consisting of layers such that, in each layer, all nodes in $H_0 \cup H_1$ take compute steps in round-robin order and then perform receive steps of all pending messages. 
We construct an infinite schedule by concatenating the schedules $\sigma_0' \sigma_1' \sigma_2$ in the natural way; we refer the reader to \cite{lynch:1996} for the formal definitions of concatenating schedules.
It is straightforward to verify that $\sigma_0' \sigma_1' \sigma_2$ results in an admissible execution for the clique network $G$ according to the asynchronous broadcast model. 

To conclude our proof, we use an indistinguishability argument.
For a given network $H$, let $\vec{r}$ be a vector of $|H|$ bit-strings, representing the respective sequences of random coin flips observed by the nodes in $H$.
We define $\alpha(H,\vec{r},N,\sigma)$ to be the execution where nodes in $H$ observe the coin flips given by $\vec{r}$, have knowledge of the network size estimate $N$, and execute steps according to some schedule $\sigma$. 
Note that $\alpha(H,\vec{r},N,\sigma)$ is an execution prefix if $\sigma$ is finite.
By construction, all messages between $H_0$ and $H_1$ are still pending for delivery at the end of schedule $\sigma_0'\sigma_1'$.
It follows that, for any vector of random strings $\vec{r}$, the execution prefixes $\alpha(G,\vec{r},2n,\sigma_0')$ and $\alpha(H_0,\vec{r},2n,\sigma_0')$ are indistinguishable for nodes in $H_0$, i.e., they perform the same sequence of local state transitions. Similarly,  $\alpha(G,\vec{r},2n,\sigma_0'\sigma_1')$ and $\alpha(H_1,\vec{r},\sigma_1')$ are indistinguishable for nodes in $H_1$.

Recall that the lengths of the prefixes $\sigma_0'$ and $\sigma_1'$ are chosen in a way such that all nodes in $H_0$ (resp.\ $H_1$) decide in the (finite) schedule $\sigma_0'$ (resp.\ $\sigma_1'$) with probability $\ge 1 - 1/n$, and by the above indistinguishability, the same is true by the end of schedule $\sigma_0'\sigma_1'$.
Conditioned on the event $E$ that this is happens, we have argued above that \emph{all} nodes in $H_0$ decide on $0$  with probability at least $1 - \epsilon$ when executing the schedule $\sigma_0'\sigma_1'\sigma_2$ in the network $G$. 
Given the same schedule, nodes in $H_1$ decide on $1$ with probability $\ge 1 - \epsilon$ and hence agreement is violated with probability at least $(1 - \epsilon)^2$.
Let $F$ be the event that the algorithm fails. 
Since we have assumed that the algorithm fails with probability at most $\epsilon$, we get
\[
  \epsilon \ge \Prob{ F }  
           \ge \Prob{F \mid E} \Prob{E} 
           \ge \left(1 - \epsilon\right)^2\left(1 - \tfrac{1}{n}\right)^2 
           \ge \tfrac{1}{2}\left(1 - \epsilon\right)^2.
\]
Solving the inequality yields $\epsilon\ge 2 - \sqrt{3}$ as required.
\end{proof}
}

\bibliographystyle{plainurl}
\bibliography{wireless-consensus,wireless,sinr}

\end{document}